\documentclass[11pt]{article}

\usepackage{amsmath,amssymb,graphicx,color,amsthm}
\usepackage[linesnumbered,ruled]{algorithm2e}
\usepackage{tikz}
\usepackage{dsfont}
\usepackage{mathtools}
\usetikzlibrary{shapes.geometric, arrows}
\usepackage{algorithmic}
\usepackage[english]{babel}
\usepackage[utf8]{inputenc}
\usepackage{mathrsfs}  
\usepackage{afterpage}
\usepackage{mathtools}
\usepackage{amsfonts}
\usepackage[colorinlistoftodos]{todonotes}
\usepackage{enumerate}
\usepackage[round]{natbib}
\usepackage{hyperref}
\newcommand{\excise}[1]{}

\newcommand\RR{\mathbb{R}}

\newcommand\EE{\mathbb{E}}

\newcommand{\tr}{\operatorname{tr}}

\newcommand\blfootnote[1]{%
  \begingroup
  \renewcommand\thefootnote{}\footnote{#1}%
  \addtocounter{footnote}{-1}%
  \endgroup
}

\newtheorem{theorem}{Theorem}
\newtheorem{definition}{Definition}
\newtheorem{lemma}{Lemma}

\newtheorem*{example*}{Example}

\newtheorem{corollary}{Corollary}
\newtheorem{remark}{Remark}

\DeclareMathOperator\var{var}
\DeclareMathOperator\diag{diag}

\DeclareMathOperator\argmin{argmin}
\DeclareMathOperator\argmax{argmax}

\DeclareMathOperator\Id{I}
\DeclareMathOperator\eig{eig}
\DeclareMathOperator{\Gr}{Gr}
\DeclarePairedDelimiterX{\infdivx}[2]{(}{)}{%
	#1\;\delimsize\|\;#2%
}

\newcommand{\RNum}[1]{\uppercase\expandafter{\romannumeral #1\relax}}

\xdefinecolor{dukeblue}{rgb}{0.004,0.129,0.412}

\begin{document}

\title{Probabilistic Contrastive Principal Component Analysis}
\author{\\[1ex]Didong Li$^{1,2}$\thanks{Equal contribution}, Andrew Jones$^{1}$\footnotemark[1], and Barbara E. Engelhardt$^{1,3}$\\
{\em Department of Computer Science, Princeton University$^{1}$}\\
{\em Department of Biostatistics, University of California, Los Angeles$^{2}$}\\
{\em Center for Statistics and Machine Learning, Princeton University$^{3}$}

}
\maketitle

\begin{abstract}
Dimension reduction is useful for exploratory data analysis. In many applications, it is of interest to discover variation that is enriched in a ``foreground'' dataset relative to a ``background'' dataset. Recently, contrastive principal component analysis (CPCA) was proposed for this setting. However, the lack of a formal probabilistic model makes it difficult to reason about CPCA and to tune its hyperparameter. In this work, we propose probabilistic contrastive principal component analysis (PCPCA), a model-based alternative to CPCA. We discuss how to set the hyperparameter in theory and in practice, and we show several of PCPCA's advantages over CPCA, including greater interpretability, uncertainty quantification and principled inference, robustness to noise and missing data, and the ability to generate data from the model. 
We demonstrate PCPCA's performance through a series of simulations and case-control experiments with datasets of gene expression, protein expression, and images.\blfootnote{Code for the model and experiments is available at \url{https://github.com/andrewcharlesjones/pcpca}.}
\end{abstract}

\section{Introduction}
Principal component analysis (PCA) is a popular technique for dimension reduction and data visualization~\citep{hotelling1933analysis}. PCA has been widely used to understand the low-dimensional structure of datasets in a variety of scientific applications~\citep{jirsa1994theoretical,brenner2000adaptive,novembre2008interpreting, darbyshire2016pricing,pasini2017principal}.
In addition to its practical utility in data exploration tasks, estimation in PCA is computationally feasible using, for example, singular value decomposition (SVD). Moreover, PCA offers a satisfying geometric interpretation, namely, that the PCs capture orthogonal directions of maximum variation in the data. There is an immense literature on non-linear generalizations of PCA including kernel PCA~\citep{scholkopf1998nonlinear}, generalized PCA~\citep{vidal2005generalized}, and principal curves~\citep{hastie1989principal}, as well as modifications to PCA that incorporate sparsity~\citep{tibshirani1996regression,zou2005regularization,zou2006sparse}, robustness~\citep{candes2011robust}, and more. In addition, probabilistic PCA (PPCA, \citealt{roweis1998algorithms,tipping1999probabilistic}) was developed to provide a model-based alternative to PCA, where the traditional objective function is re-interpreted as the likelihood estimate of a latent variable model that is a special homoskedastic version of Gaussian factor analysis~\citep{fruchter1954introduction}. A non-linear version of probabilistic PCA was described soon afterwards in a Gaussian process latent variable model (GPLVM, \citealt{lawrence2003gaussian}).

However vast, these collective PCA methods are still not suitable for some applications. In this work, we consider settings in which the dataset consists of two groups --- a \emph{foreground group} and a \emph{background group} --- and we are interested in identifying structure, variation, and information unique to the foreground group. This situation arises naturally in many scientific experiments with two or more subpopulations, such as case-control studies. For example, in a genomics context, the foreground data could be gene expression measurements from patients with a disease, and the background data could be measurements from healthy patients~\citep{twine2011whole,zheng2017massively,young2018single}. In this case, the goal is to identify transcriptional structure that is enriched in patients with the disease relative to healthy patients. Clearly, PCA is not suitable in this contrastive setting because PCA only identifies structure that exists across the union of the two groups or structure in each group in isolation. 

Contrastive modeling approaches have recently been proposed for this purpose. As a first push in this direction, a general contrastive learning framework was developed for mixture models~\citep{zou2013contrastive}. More recently, contrastive PCA (CPCA) was developed~\citep{abid2017contrastive,abid2018exploring} to find contrastive principal components (CPCs) that maximize variance in the foreground and minimize the variance in the background. However, in its original formulation, CPCA lacks a formal probabilistic model, so it is difficult to perform statistical inference within this framework. Moreover, the current CPCA framework does not allow a geometric interpretation.

In this paper, we develop probabilistic contrastive principal component analysis (PCPCA), a model-based alternative to CPCA for contrastive variation estimation. We recast the CPCA objective in a way that is amenable to a geometric interpretation, and we extend this analysis to the probabilistic setting. We then present a novel contrastive objective function which takes the form of a relative likelihood, and we provide a simple maximum relative likelihood estimate (MRLE) for the model. Furthermore, we develop a gradient descent algorithm that optimizes the objective in the presence of missing data. 

We show that PCPCA is a more general model than PCA, PPCA, or CPCA, and that these three methods can be recovered as special cases of PCPCA, thus providing a unifying framework to understand these methods. Unlike CPCA, our model is both generative, providing a model-based approach that allows for uncertainty quantification and principled inference. Unlike PPCA, our proposed method extracts variation that is unique to the foreground data while excluding variation shared between the foreground and background data, which is a critical goal in many experimental settings. 

PCPCA may be applied to a variety of statistical and machine learning problem domains including dimension reduction, synthetic data generation, missing data imputation, and clustering. We demonstrate the model's behavior and capabilities through an extensive series of simulations and experiments with datasets of case/control gene and protein expression, and biological image data. 

The specific contributions of our work to this field of PCA-based methods are the following. First, we present probabilistic contrastive component analysis (PCPCA), a model-based alternative to CPCA. Next, we show that three existing dimension reduction methods --- PCA, PPCA, and CPCA --- are special cases of our model. Then, we demonstrate several advantages of PCPCA, including a more principled probability model, a geometric interpretation analogous to that of PCA, a generalized inference procedure, robustness to missing data, and the ability to generate data from the model. Finally, we provide theoretical insight into the tuning parameter $\gamma$ in both CPCA and PCPCA, which controls the degree to which the model focuses on variation in the background or foreground data.

This paper is organized as follows. First, we review related dimension reduction methods, including PCA, PPCA, and CPCA. Second, we provide a novel geometric interpretation of CPCA, along with conditions under which CPCA is well-defined. Third, we present PCPCA, derive its maximum likelihood estimators, and show that PCA, PPCA, and CPCA are special cases of this model. Fourth, we present a generalized Bayes approach for performing posterior inference. Fifth, we present a gradient descent algorithm for fitting our model in the presence of latent variables or missing data. Finally, we demonstrate our model's performance through a series of experiments with simulated, biomedical, and image data. Proofs are in the Appendix.

\section{Background}
\subsection{Principal Component Analysis (PCA)}
Let $x_1,\cdots,x_n\in\RR^D$ be i.i.d. observations and $X\in\RR^{n\times D}$ with the $i$th row $x_i^\top$. PCA is designed to find the best $d$-dimensional affine subspace to represent the data, where $d<D$. There are several equivalent definitions of PCA. We review two of them below.

The first definition is derived from a geometric perspective, where PCA finds a hyperplane $V$ that minimizes the distance between the samples and this hyperplane:
\begin{equation}\label{eqn:PCAl2}
    \underset{V^\top V=\Id_d}{\argmin} ~\sum_{i=1}^n d^2(x_i,V)=\underset{V^\top V=\Id_d}{\argmin}~ \sum_{i=1}^n \|x_i-VV^\top x_i\|^2=\underset{V^\top V=\Id_d}{\argmin}~ \frac{1}{n}\sum_{i=1}^n \|x_i-VV^\top x_i\|^2,
\end{equation}
where $V\in\RR^{D\times d}$ has orthonormal column(s) , representing a $d$-dimensional subspace of $\RR^D$. The solution is given by
\[
V=[v_1,\cdots,v_d],~v_j=\mathrm{eig}_j\left(\sum_{i=1}^nx_ix_i^\top\right)=\mathrm{eig}_j(C),
\]
where $\eig_j$ is the $j$th eigenvalue of $C=\sum_{i=1}^nx_ix_i^\top$ in the descending order.

The second definition, which leads to an equivalent solution as Equation~\eqref{eqn:PCAl2}, is motivated from a statistical perspective. In particular, PCA maximizes the variance of the projected data onto each principal component, subject to the components being orthogonal to one another (assume $d=2$ for simplicity):
\begin{equation}\label{eqn:PCAvar}
    \max_{v_1^\top v_1=1}\var(v_1^\top x_i)=\max_{v_1^\top v_1=1}~ \sum_{i=1}^n v_1^\top x_ix_i^\top v_1=\max_{v_1^\top v_1=1}~v_1^\top Cv_1,
\end{equation}
\begin{equation}
    \max_{v_2^\top v_2=1,v_1^\top v_2=0}\var(v_2^\top(x_i-v_1v_1^\top x_i))=\max_{v_2^\top v_2=1,v_1^\top v_2=0}v_2^\top Cv_2.
\end{equation}
The geometric and statistical frameworks for PCA yield equivalent solutions, but having multiple perspectives gives greater insight into the method. Our work is motivated by these complementary perspectives (\autoref{thm:CPCAloss}).

Note that we drop the mean parameter since, in practice, the data can easily be centered to have zero mean. Thus, for simplicity, throughout this paper we assume all data include features that are centered at zero.
\subsection{Probabilistic PCA (PPCA)}
PCA may be generalized in the form of a probabilistic model. Assume $z\sim N(0,\Id_d)$, $x=Wz+\epsilon$ with $W\in\RR^{D\times d}$, $\epsilon\sim N(0,\sigma^2\Id_D)$. Then
\begin{equation*}
x\sim N(0, WW^\top+\sigma^2 \Id_D).
\end{equation*}
The objective is to maximize the likelihood with respect to the parameters $W$ and $\sigma^2$: 
\begin{equation}\label{eqn:PPCA}
\underset{W,\sigma^2}{\argmax} ~p(X|W,\sigma^2).
\end{equation}
The MLE of $W$ and $\sigma^2$ are given by \citep{roweis1998algorithms,tipping1999probabilistic}:
\begin{equation*}
\widehat{W}_{ML} = U(\Lambda-\widehat{\sigma}^2_{ML} I_d)^{1/2}R,~~\widehat{\sigma}^2_{ML} = \frac{1}{D-d}\sum_{i=d+1}^D\lambda_i,
\end{equation*}
where $U$ consists of the first $d$ eigenvectors of $C=\sum_{i=1}^nx_ix_i^\top$ with eigenvalues $\lambda_1\geq \cdots\geq\lambda_D>0$, $\Lambda=\diag\{\lambda_1,\cdots,\lambda_d\}$ and $R\in \mathrm{O}(d)$ is any rotation matrix.

That is, the hyperplane obtained by PPCA only differs by a re-scaling of the basis from the PCA hyperplane. In other words, PPCA is ``equivalent'' to PCA, and this becomes exact when $\sigma^2\to0$.

\subsection{Contrastive PCA (CPCA)}
PCA can also be generalized for contrastive modeling of two datasets. For foreground observations $x_1,\cdots,x_n\in\RR^D$ and background observations $y_1,\cdots,y_m\in\RR^D$, contrastive PCA (CPCA, \citealt{abid2018exploring}) is designed to discover low-dimensional structure that
is unique to or enriched in the foreground dataset $X$ relative to the background dataset $Y$. Let $C_X = \frac{1}{n}\sum_{i=1}^n x_ix_i^\top$ be the empirical covariance matrix for $X$ and $C_Y = \frac{1}{m}\sum_{j=1}^my_jy_j^\top$ for $Y$. 

Recall the statistical perspective of PCA given by \autoref{eqn:PCAvar}. In the contrastive setting, for any unit vector $v$, we have two variances --- the foreground variance and background variance --- given by $v^\top C_Xv$ and $v^\top C_Yv$, respectively. The objective of CPCA is to identify directions $v$ that account for a large amount of variance in the foreground and a small amount of variance in the background. Specifically, CPCA solves the following optimization problem:
\begin{equation}\label{eqn:CPCAvar}
    \underset{v^\top v=1}{\argmax}~v^\top C_Xv-\gamma v^\top C_Yv=\underset{v^\top v=1}{\argmax}~v^\top Cv,
\end{equation}
where $\gamma\in[0,\infty]$ is a tuning parameter, and $C=C_X-\gamma C_Y$.

The solution of CPCA is the same as PCA if we replace $C_X$ by $C$, namely, the optimal $v$ is the top eigenvector of $C$. From the definition of CPCA, it is clear that CPCA reduces to PCA when $\gamma =0$.

\section{A deeper look at CPCA}
In this section, we analyze some important aspects of CPCA that were not discussed in previous studies. These analyses provide insight into when CPCA is well-defined, and in turn provide motivation for our proposed model, PCPCA, which is described in the next section.

\subsection{Geometric interpretation of CPCA}
CPCA was originally defined from a statistical perspective (\autoref{eqn:CPCAvar}, \citealt{abid2018exploring}). Recalling the geometric definition of PCA, \autoref{eqn:PCAl2}, it is natural to consider whether there also exists a geometric interpretation for CPCA.
\begin{theorem}\label{thm:CPCAloss}
The statistical objective function of CPCA in Equation \eqref{eqn:CPCAvar} is equivalent to the following geometric objective function
\begin{equation}\label{eqn:CPCAl2}
    \underset{v^\top v}{\argmin}~\frac{1}{n} \sum_{i=1}^n \|x_i-vv^\top x_i\|^2-\gamma \frac{1}{m} \sum_{j=1}^m \|y_j-vv^\top y_j\|^2.
\end{equation}
\end{theorem}

\noindent The proof can be found in Appendix \ref{pf:CPCAloss}. From a geometric perspective, the objective of CPCA is to find a hyperplane that is close to the foreground data but far from the background data. This coincides with the intuition of CPCA's overall goal, which is to identify the unique information in the foreground data.

In addition to the geometric intuition provided by \autoref{thm:CPCAloss}, this theorem also allows us to adapt distance-based algorithms to the contrastive setting. Specifically, in any of these algorithms, one can consider replacing the distance by the ``constrastive distance.'' For example, sparse CPCA has been developed following this philosophy \citep{boileau2020exploring}, although without this justification. However, it is important to note that the contrastive distance is not a well-defined distance, which may violate the assumptions of traditional distance-based algorithms, and so cannot be used to replace distance metrics in existing algorithms without some luck.

\subsection{Positive definiteness of $C$}
Another crucial consideration of CPCA is the positive definiteness of $C$, which may be treated as a ``covariance matrix.'' However, $C$ is not necessarily positive definite unless $\gamma =0$, in which case $C=C_X$. 

Here, we derive a sufficient condition on $\gamma$ such that $C$ is positive definite. Let the eigenvalues of $C$, $C_X$, and $C_Y$ be $\lambda_1\geq \cdots\lambda_D\geq0$, $\mu_1\geq \cdots\mu_D\geq0$ and $\rho_1\geq\cdots\geq\rho_D\geq 0$, respectively.
\begin{lemma}\label{lem:PD}
$C$ is positive definite if
\begin{equation*}
    \gamma < \frac{\min\{\mu_1,\cdots,\mu_D\}}{\max\{\rho_1,\cdots,\rho_D\}}.
\end{equation*}
\end{lemma}
The proof can be found in Appendix \ref{pf:PD}.
However, in CPCA, the positive definiteness of $C$ is not strictly required since the target is a $d\ll D$ dimensional subspace, and $D$ is large for high-dimensional data, such as biomedical data. Instead, CPCA only requires that the first $d$ eigenvalues of $C$ must be positive. In many applications for visualization and clustering, $d=2$~\citep{abid2018exploring}, which allows $\gamma$ to be defined over a wide range.

The following theorem provides a necessary and sufficient condition for the first $d$ eigenvalues of $C$ being positive, with Lemma \ref{lem:PD} as a special case when $d=D$.
\begin{theorem}\label{thm:Weyl}
The first $d$ eigenvalues of $C$ are positive if
\begin{equation*}
\gamma < \max\left\{\frac{\mu_d}{\rho_{1}},\frac{\mu_{d+1}}{\rho_{2}},\cdots,\frac{\mu_D}{\rho_{D-d+1}}\right\}.
\end{equation*}
Otherwise, there exists a $C$ such that the $d$th eigenvalue is negative. That is, the upper bound is tight.
\end{theorem}
The proof can be found in Appendix \ref{pf:Weyl}.

\begin{corollary}\label{cly:loss}
For a fixed $d$, a larger $\gamma$ corresponds to a smaller loss. For a fixed $\gamma$, the loss will decrease when $d$ is increased to $d+1$ if $\gamma<\max\left\{\frac{\mu_{d+1}}{\rho_{1}},\frac{\mu_{d+2}}{\rho_{2}},\cdots,\frac{\mu_D}{\rho_{D-d}}\right\}.$
\end{corollary}
The proof can be found in Appendix \ref{pf:loss}. The above corollary explains why, when $\gamma$ is large, a smaller $d$ is preferable: when $\gamma$ is large enough (such that $\lambda_{3}<0$), CPCA with $d=2$ is better than higher dimensional CPCA in terms of mean squared error (MSE). See Section \ref{sec:experiments} for more details.

\subsection{The tuning parameter $\gamma$}
In CPCA, the tuning parameter $\gamma$ can be any non-negative real number, making it difficult to tune. Although a tuning method was suggested in the original CPCA proposal~\citep{abid2018exploring}, the procedure depends on an almost exhaustive search, making it inefficient. We first analyze the role of $\gamma$ and propose a new parameterization such that the new tuning parameter $\gamma$ always lies in a small range, typically close to $[0,1]$, making it easier to tune.

Recall that, for PCA, minimizing the sum of squared error and minimizing the mean squared error are equivalent, since the scale $\frac{1}{n}$ only changes the eigenvalues of the sample covariance, not its eigenvectors:
\begin{equation*}
    \underset{v^\top v=1}{\argmin} \sum_{i=1}^n \|x_i-vv^\top x_i\|^2=\underset{v^\top v=1}{\argmin}~\frac{1}{n} \sum_{i=1}^n \|x_i-vv^\top x_i\|^2.
\end{equation*}
However, in the contrastive setting, the scale matters. Specifically, the following two optimization problems are not equivalent unless $m=n$, which rarely happens in practice:
\begin{align}
    &\underset{v^\top v=1}{\argmin} \sum_{i=1}^n \|x_i-vv^\top x_i\|^2-\gamma' \sum_{i=1}^m \|y_i-vv^\top y_i\|^2 \label{eqn:CPCAl22}\\
    &\hspace{-0.6cm}\neq \underset{v^\top v=1}{\argmin}~ \frac{1}{n}\sum_{i=1}^n \|x_i-vv^\top x_i\|^2-\frac{1}{m}\gamma' \sum_{i=1}^m \|y_i-vv^\top y_i\|^2.\nonumber
\end{align}
Comparing Equations \eqref{eqn:CPCAl2} and \eqref{eqn:CPCAl22}, we conclude that they are equivalent when $\gamma'=\gamma\frac{n}{m}$. 
If the sample sizes of the two groups are not the same, then the choice of $\gamma$ is different from the choice of $\gamma'$. We will show that this reparameterization, which is adjusted by the relative sample size, makes $\gamma$ more interpretable and easier to tune.

\section{Probabilistic CPCA (PCPCA)}
In this section, we present a probabilistic approach to contrastive learning. First, we present adjacent work on contrastive learning performed through probabilistic modeling. Then, we present our model, PCPCA, and analyze it through the lens of CPCA, PPCA, and PCA.

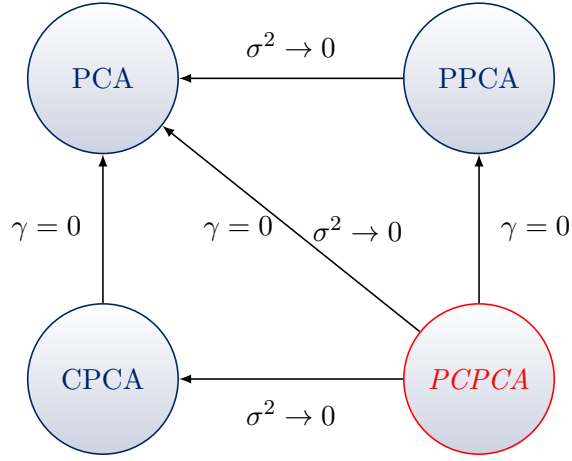
\begin{figure}[h!]
\begin {center}
\begin {tikzpicture}[-latex ,auto ,node distance =4 cm and 5cm ,on grid ,
semithick ,
state/.style ={ circle ,top color =white , bottom color = dukeblue!20 ,
draw,dukeblue , text=dukeblue , minimum width =1 cm}]
\node[state] (C)[minimum size=2cm] {CPCA};
\node[state,color=red] (D) [right = of C, minimum size=2cm,color=red]{ \emph{PCPCA}};
\node[state] (A) [above =of C,minimum size=2cm] {PCA};
\node[state] (B) [above  =of D,minimum size=2cm] {PPCA};
\path (C) edge [] node[left =0.15 cm] {$\gamma=0$} (A);
\path (B) edge [] node[above =0.15 cm] {$\sigma^2\to0$} (A);
\path (D) edge [] node[right =0.15 cm] {$\gamma=0$} (B);
\path (D) edge [] node[below =0.15 cm] {$\sigma^2\to0$} (C);
\path (D) edge [] node[right=0.15cm] {$\sigma^2\to0$} node[left=0.1cm]{$\gamma=0$} (A);
\end{tikzpicture}
\caption{Target commutative diagram for PCA family.}
\label{fig:diagram}
\end{center}
\end{figure}

\subsection{Contrastive latent variable model}

As a slightly different model than CPCA, the contrastive latent variable model (CLVM) has been proposed~\citep{severson2019unsupervised}:
\begin{equation}\label{eqn:CLVM}
    z\sim N(0,\Id_d),~t\sim N(0,\Id_{d'}),~ x=Sz+Wt+\epsilon_x,~y=Sz+\epsilon_y,~\epsilon_x,\epsilon_y\sim N(0,\sigma^2\Id_D).
\end{equation}
The marginals are given by
\begin{equation*}
    x\sim N(0,SS^\top+WW^\top+\sigma^2\Id_D),~~y\sim N(0,SS^\top+\sigma^2\Id_D).
\end{equation*}

The objective function in inference is the likelihood or log likelihood.
When the background dimension is zero, that is, when $d'=0$, the above model becomes PPCA for $X\cup Y$.

However, we are interested in characterizing the foreground data $X$ while controlling for variation in the background $Y$. For this reason, it is more desirable to recover PPCA for $X$ as a special case of the model rather than PPCA for $X \cup Y$ or $Y$. Recall that similar a relation holds for PCA and CPCA: when $\gamma =0$, CPCA becomes PCA on $X$ only. In fact, there does not exist any $\gamma$ such that CPCA is equivalent to PCA on $X\cup Y$.

In addition, there is no clear link between the CLVM (\autoref{eqn:CLVM}) and CPCA, even if $\sigma^2\to 0$. As a result, it is of interest to develop a general model that is simultaneously a probabilistic version of CPCA and a contrastive version of PPCA.

\subsection{PCPCA model}
Consider the following model 
\begin{equation}
z_x,z_y\sim N(0,\Id_d),~x=Wz_x+\epsilon_x,~y=Wz_y+\epsilon_y,
\end{equation}
where $\epsilon_x,~\epsilon_y\sim N(0,\sigma^2\Id_D)$ are i.i.d. Gaussian noise vectors. Recall the equivalent statistical and geometric interpretations of PPCA: maximizing the likelihood of $X$ is equivalent to minimizing the distance between $X$ and the hyperplane $W$.
For CPCA, we expect such $W$ to be far away from the background $Y$, which is exactly the (geometric) objective of CPCA (see \autoref{thm:CPCAloss}). 
In the probabilistic setting, maximizing the distance from $Y$ is equivalent to minimizing the likelihood of $Y$. This is counter-intuitive, but it coincides with our model's motivation to account for variation in the foreground data, not the background data. 
Thus, we have the following objective function
\begin{equation}\label{eqn:PCPCA}
\underset{W,\sigma^2}{\argmax} ~\frac{ p(X|W,\sigma^2)}{p(Y|W,\sigma^2)^\gamma}.
\end{equation}
The above objective function becomes the PPCA objective function when $\gamma=0$, and a relative likelihood when $\gamma=1$. For general $\gamma\in[0,\infty)$, we refer to Equation \ref{eqn:PCPCA} as the relative likelihood, as it captures the likelihood of the foreground data with respect to the likelihood of the background data, scaled by $\gamma$. 

The non-traditional nature of this objective requires further comment. Notice that this objective is not a traditional likelihood ratio, which is typically defined as a ratio of the likelihood of one dataset under two different parameter settings. Rather, ours is a ratio of likelihoods of two different datasets under a shared parameter setting. Furthermore, Equation \eqref{eqn:PCPCA} is not a well-defined likelihood unless $\gamma=0$. These caveats preclude the use of traditional estimation and inference procedures based on likelihoods and relative likelihoods. For this reason, we present alternative procedures: one based on a direct maximization of Equation \eqref{eqn:PCPCA} and another based on generalized posterior inference. 

First, we investigate the closed-form solution for this objective.
\begin{theorem}\label{thm:PCPCA}
The $W$ and $\sigma^2$ that maximize Equation \eqref{eqn:PCPCA}, denoted by $\widehat{W}_{ML}$, $\widehat{\sigma}^2_{ML}$, are given by
\begin{eqnarray*}
\widehat{\sigma}^2_{ML} &=&\frac{1}{(n-\gamma m )(D-d)}\sum_{i=d+1}^D\lambda_i\\
    \widehat{W}_{ML} &=& U_d\left(\frac{\Lambda_d}{n-\gamma m}-\widehat{\sigma}^2_{ML}\Id_d\right)^{1/2}R,
\end{eqnarray*}
where $U_d$ consists of the first $d$ eigenvectors of $C=\sum_{i=1}^nx_ix_i^\top-\gamma \sum_{j=1}^m y_jy_j^\top$, $\Lambda_d=\diag\{\lambda_1,\cdots,\lambda_d\}$ contains the corresponding eigenvalues, and $R$ is any $d$ by $d$ rotation matrix. Moreover, $\widehat{W}_{ML}$ is equivalent to the $\widehat{W}_{PPCA}$ that maximizes the PPCA objective when $\gamma=0$, and is equivalent to the $\widehat{W}_{CPCA}$ that maximizes the CPCA objective as $\sigma^2\to0$.
\end{theorem}
The proof can be found in Appendix \ref{pf:PCPCA}. As a result, we find last missing piece in the commutative diagram \ref{fig:diagram}, namely, we have a complete algorithm for PCPCA that allows PCA, PPCA, and CPCA to be recovered as subcases of this general framework. 

\begin{remark}
The above PCPCA solution highlights two hidden assumptions for PCPCA:
\begin{enumerate}
    \item $n-\gamma m>0$ so that $\widehat{W}_{ML}$ is well defined, that is, $\gamma <\frac{n}{m}$.
    \item $\sum_{i=d+1}^D\lambda_i>0$ so that $\widehat{\sigma}^2_{ML}>0$. By \autoref{thm:Weyl}, a sufficient condition is $\gamma<\frac{\sum_{i=d+1}^D\mu_i}{(D-d)\rho_1}$.
\end{enumerate}
These seemingly strong constraints restrict the range of $\gamma$ from $[0,\infty)$ to a small interval, often a subinterval of $[0,1]$. This more restricted interval makes PCPCA easier to tune than CPCA. In addition, the performance of PCPCA is robust to the choice of $\gamma$ within this interval, which is not observed for CPCA. 
\end{remark}

\section{Generalized Bayesian approach}
Next, we present a generalized Bayesian framework for performing posterior inference in the PCPCA model. Recall that our objective (\autoref{eqn:PCPCA}) is not a likelihood, so we cannot simply place a prior on $\theta = (W, \sigma^2)$ and perform Bayesian inference in the traditional fashion. For this reason, we leverage more general loss-based inference methods based on Gibbs posteriors.

\subsection{Gibbs Posterior}
Let $\Theta$ be the space of parameters, which can be a finite or infinite dimensional space, and $U$ be the feature space, then we denote the loss function $l: U\times \Theta\to \RR$ and $l_\theta:U\to \RR$. For a given measure $P$ on $U$ (often the true measure), define the risk function $R:\Theta\to \RR$ to be 
$$R:\Theta\to \RR,~\theta\mapsto\EE_P[l_\theta].$$ 
Then the goal is to minimize the risk: $\min_{\theta\in\Theta} R(\theta)$. It is common to assume the minimizer is unique, denoted by $\theta^*=\underset{\theta\in\Theta}{\arg\min}~R(\theta)$.

However, $P$ is often unknown. Instead, we have observations $u_1,\cdots,u_n$ and we have the corresponding empirical measure $P_n=\frac{1}{n}\sum_{i=1}^n\delta_{u_i}$. Then the empirical risk function is
$$R_n:\Theta\to \RR,~\theta\mapsto \EE_{P_n}[l_\theta].$$
The goal in this more practical setting is to minimize the empirical risk:
\[\min_{\theta\in\Theta} R_n(\theta).\]
Let $\widehat{\theta}_n=\underset{\theta\in\Theta}{\arg\min}~R_n(\theta)$ be the unique minimizer of the empirical risk. 

If a statistical model exists with density function $p_\theta$ and the loss function is $l_\theta(u)=-\log p_\theta(u)$, then this minimization reduces to maximum likelihood estimation. In particular, in this setting, $R_n(\theta)$ is the negative log likelihood and $\widehat{\theta}_n$ is the MLE.
\begin{definition}
Given a prior $\Pi$ on $\Theta$, the Gibbs posterior is defined as
$$\Pi_n(d\theta)\propto e^{-wnR_n(\theta)}\Pi(d\theta),~\theta\in\Theta,$$
where $\Pi$ is the prior and $w>0$ is the learning rate.
\end{definition}
The Gibbs posterior becomes the true posterior if $R_n$ is the negative log-likelihood. 

\begin{definition}
The Gibbs posterior $\Pi_n$ asymptotically concentrates around $\theta^*$ at the rate $\varepsilon_n\to0$ w.r.t. divergence measure $d$ on $\Theta$ if
\[ \EE_{P_n}\Pi_n\left(\{\theta:d(\theta,\theta^*)>M\varepsilon_n\}\right)\to 0, \]
where $M$ is a constant.
\end{definition}
Note that $d$ is only required to be positive semi-definite, that is, $d(\theta;\theta')\geq 0$ with equality iff $\theta=\theta'$. 

\subsection{Gibbs posterior for CPCA}
We next propose a loss function and Gibbs posterior for CPCA. Let $u=(x,\alpha)$ where $x\in\RR^D$ is the observation, $\alpha\in\{0,1\}$ with $\alpha=0$ represents foreground data while $\alpha=1$ represents background data and $\theta=V\in\Gr(D,d)$, the Grassmannian manifold. Let the loss function be $$l_\theta(u) = (-\gamma)^{\alpha}d^2(x,V)=(-\gamma)^{\alpha}\|x-VV^\top x\|^2=(-\gamma)^{\alpha}\left(x^\top x-xVV^\top x\right).$$ For simplicity, assume $u_1,\cdots,u_n$ are foreground data while $u_{n+1},\cdots,u_{n+m}$ are background data. As a result, the empirical risk function is
\begin{align*}
R_n(\theta)&=\frac{1}{n+m}\sum_{i=1}^{n+m} l_\theta(u_i)\\
&=\frac{1}{n+m} \left(\sum_{i=1}^n \left(x_i^\top x_i-x_i^\top VV^\top x_i\right)-\gamma\sum_{j=1}^m \left( x_{n+j}^\top x_{n+j}-x_{n+j}^\top VV^\top x_{n+j}\right)\right)\\
&= -\sum_{i=1}^n x_i^\top VV^\top x_i+\gamma\sum_{j=1}^m x_{n+j}^\top VV^\top x_{n+j}+M\\
&= -\tr(VV^\top C)+M,
\end{align*}
where $C=\sum_{i=1}^nx_ix_i^\top-\gamma \sum_{j=1}^mx_{n+j}x_{n+j}^\top$ and $M$ is independent of $\theta$. 
We conclude that the empirical risk function coincides with the objective function of CPCA. Furthermore, if the prior $\Pi$ is chosen to be the uniform prior, then the maximum a posteriori estimation (MAP) matches the solution of CPCA, which is the subspace spanned by the first $d$ eigenvectors of $C$.  

For the population version of the risk, assume $P=\beta P_F+(1-\beta)P_B$ where $\beta\in(0,1)$, $P_F$ is the foreground measure with zero mean and covariance $C_F$, and $P_B$ is the background measure with zero mean and covariance $C_B$. Then the risk function is
\begin{align*}
R(\theta)&=\EE_Pl_\theta(u)\\
&=\beta \EE_{x\sim P_F} \left(x^\top x-x^\top VV^\top x\right)-(1-\beta)\gamma\EE_{x\sim P_B}\left(x^\top x-x^\top VV^\top x\right)\\
&= -\beta\tr(VV^\top C_F)-(1-\beta)\gamma\tr(VV^\top C_B)+M\\
&= -\tr(VV^\top C)+M,
\end{align*}
where $C=\beta C_F-(1-\beta)\gamma C_B$ and $M$ is independent of $\theta$. So the minimizer is given by
$$\theta^*=V^*=[\eig_1(C),\cdots,\eig_d(C)].$$

We consider the risk divergence $d(\theta;\theta^*)=(R(\theta)-R(\theta^*))^{1/2}=\tr((V^*V^{*\top}-VV^\top)C)^{1/2}$, which measures the difference between risks. We now consider the contraction rate of this Gibbs posterior.
\begin{theorem}\label{thm:CPCArate}
Assume $P=\beta N(0,C_F)+(1-\beta)N(0,C_B)$ and let the prior $\Pi$ be uniform on $\Gr(D,d)$ w.r.t. the Haar measure, then the Gibbs posterior $\Pi_n$ asymptotically contracts to $\theta^*$ w.r.t.  $d$ at rate $n^{-1/2}$.
\end{theorem}
The proof can be found in Appendix \ref{proof:CPCArate}. As a result, the Gibbs posterior will contract to the optimal parameter as the sample size increases, which provides theoretical support for the generalized Bayesian version of CPCA.

\subsection{Gibbs Posterior for PCPCA}
We now consider the Gibbs posterior for PCPCA. Let $u=(x,\alpha)$ where $x\in\RR^D$ is the observation and $\alpha\in\{0,1\}$ indicates the sample's condition, with $\alpha=0$ representing foreground data while $\alpha=1$ representing background data. As before, let $\theta=(W,\sigma^2)$ be the parameter. Let the loss function be $l_\theta(u) = -(-\gamma)^{\alpha}\log N(v;0,WW^\top+\sigma^2\Id_D)$. For simplicity, assume $u_1,\cdots,u_n$ are foreground data while $u_{n+1},\cdots,u_{n+m}$ are background data. As a result, the empirical risk function is
\begin{align*}
&R_n(\theta)=\frac{1}{n+m}\sum_{i=1}^{n+m} l_\theta(u_i)\\
&=\frac{1}{n+m} \left(\sum_{i=1}^n-\log N(x_i;0,WW^\top+\sigma^2\Id_d)+\gamma\sum_{j=1}^m \log N(x_{n+j};0,WW^\top+\sigma^2\Id_d)\right)\\
&= \sum_{i=1}^n \left(\frac{1}{2}\log |A|+\frac{1}{2}x_{i}^\top A^{-1}x_i\right)+\gamma \sum_{j=1}^m\left(-\frac{1}{2}\log |A|-\frac{1}{2}x_{n+j}^\top A^{-1}x_{n+j}\right)+M\\
&= \frac{n-\gamma m}{2}\log |A|+\frac{1}{2}\tr(A^{-1}C)+M,
\end{align*}
where $A = WW^\top+\sigma^2 \Id_D$, $C=\sum_{i=1}^nx_ix_i^\top-\gamma \sum_{j=1}^mx_{n+j}x_{n+j}^\top$ and $M$ is independent of $\theta$. 
We conclude that the empirical risk function coincides with the (negative log) objective function of PCPCA. Furthermore, if the prior $\Pi$ is chosen to be the uniform prior, then the maximum a posteriori (MAP) estimate matches the solution in \autoref{thm:PCPCA}. 

For the population version of the risk, assume the same model as in the previous section. Specifically, we assume $P=\beta P_F+(1-\beta)P_B$ where $\beta\in(0,1)$, $P_F$ is the foreground measure with zero mean and covariance $C_F$, and $P_B$ is the background measure with zero mean and covariance $C_B$. Then the risk function is
\begin{align}
R(\theta)&=\EE_Pl_\theta(u) \nonumber \\
&=-\beta \EE_{x\sim P_F} \log p(x|0,A)+(1-\beta)\gamma\EE_{x\sim P_B}\log p(x|0,A) \nonumber \\
&= \frac{\beta}{2}\left(\log|A|+\tr(A^{-1}C_F)\right)-\frac{(1-\beta)}{2}\gamma\left(\log|A|+\tr(A^{-1}C_B)\right)+M \nonumber \\
&= \frac{\beta-(1-\beta)\gamma}{2}\log|A|+\frac{1}{2}\tr(A^{-1}C)+M, \label{eq:pcpca_risk}
\end{align}
where $C=\beta C_F-(1-\beta)\gamma C_B$. So the minimizer is given by
\begin{equation}
    {\sigma^{*2}}=\frac{1}{D-d}\sum_{i=d+1}^D\lambda_i,W^*=U_d\left(\frac{\Lambda_d}{\beta-(1-\beta)\gamma}-\sigma^{*2}\Id\right)^{1/2}R \label{eq:risk_minizer}
\end{equation}
where $\Lambda_d=\diag\{\lambda_i\}$ consists of the largest $d$ eigenvalues of $C$, and $U_d$ consists of the corresponding $d$ eigenvectors. 
We consider the same risk divergence as in the previous section: 
$$d(\theta;\theta^*)=(R(\theta)-R(\theta^*))^{1/2}=\left(\frac{\beta-(1-\beta)\gamma}{2}(\log|A|-\log|A^*|)+\frac{\tr((A^{-1}-A^{*-1})C)}{2}\right)^{1/2}.$$
We now consider the contraction rate of the PCPCA Gibbs posterior under this divergence.

\begin{theorem}\label{thm:PCPCArate}
Assume $P=\beta N(0,C_F)+(1-\beta)N(0,C_B)$ and $\sigma^2\geq \sigma_0^2>0$.  Let the prior $\Pi$ be uniform on $\RR^{D\times d}\times [\sigma_0^2,\infty)$, then the Gibbs posterior $\Pi_n$ asymptotically contracts to $\theta^*$ w.r.t. $d$ at rate $n^{-1/2}$.
\end{theorem}

The proof can be found in Appendix \ref{proof:PCPCArate}. As a result, the Gibbs posterior contracts to the optimal parameter as the sample size increases, which supports the generalized Bayesian PCPCA. 

\section{Experiments}\label{sec:experiments}
To demonstrate the behavior and usefulness of PCPCA, we fit the model on a series of simulated, gene and protein expression, and image datasets. Note that for most plots, we refer to the sample size-adjusted hyperparameter $\gamma^\prime=\frac{m}{n}\gamma$.

\subsection{Visualizing the role of the hyperparameter $\gamma$}
First, to demonstrate the role of the hyperparameter $\gamma$ in the PCPCA model, we fit the PCPCA model on a two-dimensional simulated dataset. In this simple dataset, the foreground data contain two subgroups, each of which shares an axis of variation with the background data. In particular, we generated the foreground and background by sampling $x_i \sim \mathcal{N}(\mu_x, \Sigma)$ and $y_j \sim \mathcal{N}(0, \Sigma)$ where $\mu_x = \left(\begin{smallmatrix} 1 \\ -1 \end{smallmatrix}\right)$ for half of the foreground samples, and $\mu_x = \left(\begin{smallmatrix} -1 \\ 1 \end{smallmatrix}\right)$ for the other half. For all samples, $\Sigma = \left(\begin{smallmatrix} 2.7 & 2.6 \\ 2.6 & 2.7 \end{smallmatrix}\right)$. We set the foreground and background sample sizes to be equal, $n=m=200$. We fit the PCPCA model for $\gamma^\prime \in \{0, 0.2, 0.6, 0.9\}$, and we visualize the 1-dimensional line defined by $\widehat{W}_{ML}$, where $d=1$ and $D=2$ (\autoref{fig:pcpca_toy}).

Recall that when $\gamma^\prime = 0$, PCPCA reduces to PPCA. In this case, $W$ captures the variation that is shared between the background and foreground data  (\autoref{fig:pcpca_toy}a). At higher values of $\gamma^\prime$, we observed that PCPCA captures the variation that is unique to the foreground dataset, which divides the two foreground subgroups (\autoref{fig:pcpca_toy}d). Note that $\widehat{W}_{ML}$ rotates nearly 90 degrees to capture the direction of maximal variation unique to the foreground data when $\gamma^\prime = 0.9$ relative to the PPCA solution when $\gamma^\prime=0$. At intermediate values of $\gamma^\prime$, $\widehat{W}_{ML}$ balances between capturing the shared and foreground-specific variation (\autoref{fig:pcpca_toy}c). 
\begin{figure}[h]
\includegraphics[width=1.0\textwidth]{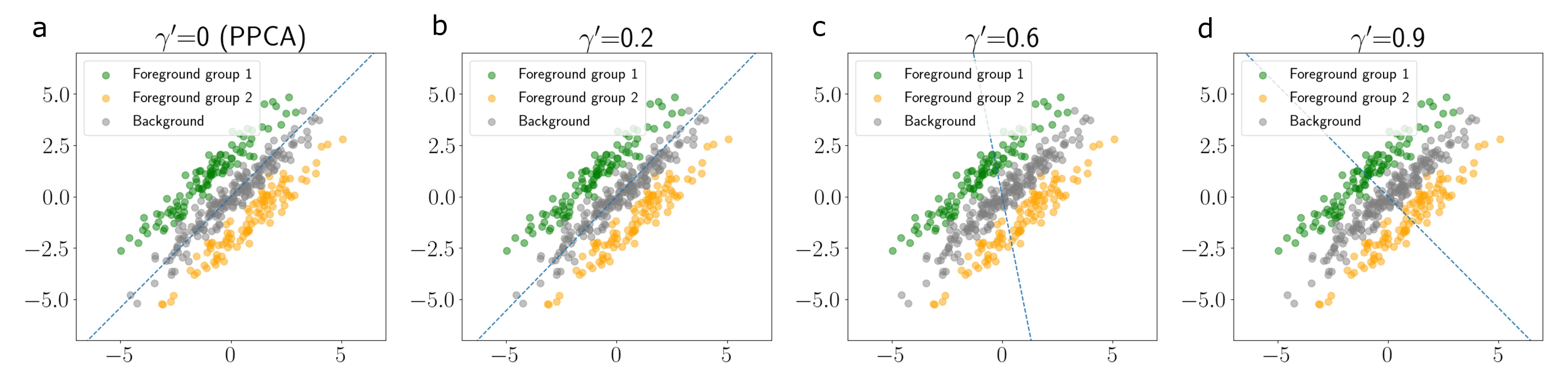}
\caption{\textbf{PCPCA MLE on simulated data.} PCPCA estimates for toy data at varying values of $\gamma^\prime$. The dotted line represents $\widehat{W}_{ML}$ ($d=1$ in this case). PCPCA recovers PPCA when $\gamma^\prime = 0$ (a) and captures the axis of variation between the two foreground subgroups when $\gamma^\prime=0.9$ (d).}
\label{fig:pcpca_toy}
\end{figure}

\subsection{Tuning $\gamma$ in experimental settings}

\subsubsection{Mouse protein expression}
We next tested PCPCA using a dataset of mouse protein expression~\citep{higuera2015self}. In this experiment, the foreground data are protein expression samples from the cortex of mice with and without Down Syndrome who were subjected to shock therapy. The background dataset consists of a set of protein expression measurements from mice without Down Syndrome who did not receive shock therapy. In total, there are $n = 270$ foreground samples and $m = 135$ background samples, each measuring the expression of $77$ proteins. The foreground samples contain $135$ mice with Down Syndrome and $135$ mice without Down Syndrome, and the intervention we model in this experiment is how shock therapy affects protein expression levels differently for mice with Down Syndrome and those without.

We fit PCPCA using a range of values for the tuning parameter $\gamma^\prime$, setting $d=2$ in each case. We found that, at higher values of $\gamma^\prime$, PCPCA was able to separate the mice with and without Down Syndrome that received shock therapy (\autoref{fig:pcpca_mice}b, c). Furthermore, PCPCA separated the background samples from the foreground samples (\autoref{fig:pcpca_mice}d). 
When $\gamma^\prime = 0$, the model is equivalent to PPCA, and visually there is minimal separation of the two groups of foreground mice (\autoref{fig:pcpca_mice}a).
\begin{figure}[h]
\centering
\includegraphics[width=1.0\textwidth]{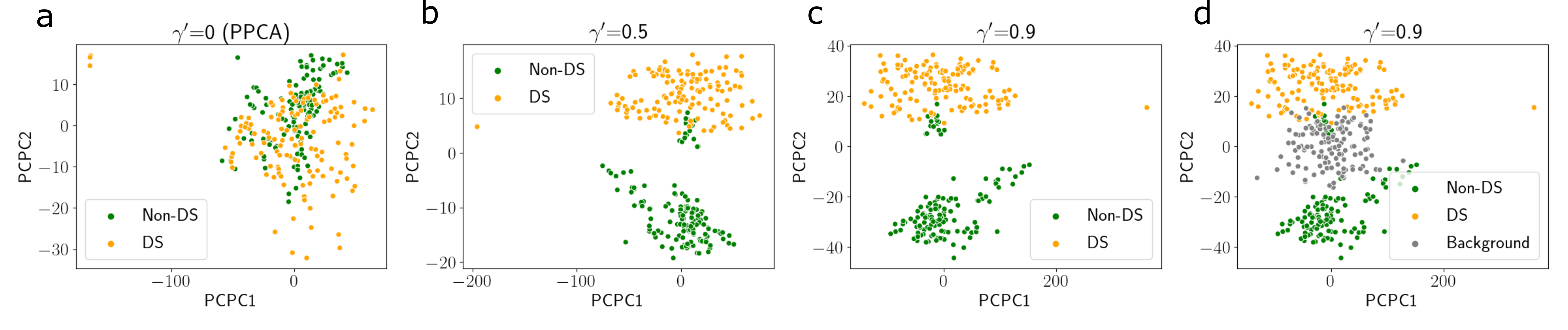}
\caption{\textbf{PCPCA on mouse protein expression data.} PCPCA applied to a dataset of mouse protein expression measurements from shock therapy-treated mice with and without Down Syndrome. Plotted in each panel are the projections of the foreground samples onto the first two components at varying values for $\gamma^\prime$. There are two subgroups in the foreground: \emph{DS} (yellow) represents mice with Down Syndrome who are exposed to shock therapy, and \emph{Non-DS} (green) represents mice without Down Syndrome who are exposed to shock therapy. (a)-(c) show the two foreground groups at each value of $\gamma^\prime$. (d) also includes the \emph{Background} dataset (gray), which is made up of mice without Down Syndrome who were not exposed to shock therapy.
}
\label{fig:pcpca_mice}
\end{figure}

We measured the degree of separation using the silhouette score (SS) of the two foreground groups of mice (Down syndrome and control) when projected into PCPCA's latent space. SS is a measure of cluster tightness~\citep{rousseeuw1987silhouettes}; higher scores represent better clustering of sample labels in the space. We found that the maximum silhouette score achieved by CPCA and PCPCA were comparable (CPCA: $0.425$, PCPCA: $0.404$). 

However, we observed different behavior between the methods in the tuning process for $\gamma^\prime$. For PCPCA, we found that SS increased monotonically with $\gamma^\prime$ (\autoref{fig:mouse_cpca_pcpca_comparison}b). In contrast, CPCA showed better clustering performance at lower values of $\gamma^\prime$, and the SS decreased with a higher $\gamma^\prime$ (\autoref{fig:mouse_cpca_pcpca_comparison}a). Additionally, the range of allowable values for $\gamma^\prime$ differed substantially between the two methods. The looser constraint on $\gamma$ in CPCA allowed for high values of $\gamma^\prime$ --- going as high as $\gamma^\prime=241$ in the mouse dataset. The reason for the large allowable values of $\gamma$ in CPCA can be understood in the context of Corollary \autoref{cly:loss}. Furthermore, at these large values of $\gamma^\prime$, the CPCA projection of the background dataset reduces to a single point (\autoref{fig:mouse_cpca_pcpca_comparison}c).  Together, these results suggest that PCPCA's parameterization allows for an easier interpretation of the tuning parameter $\gamma^\prime$, and $\gamma^\prime$ is restricted to a reasonable range in PCPCA compared with the parameter's range in CPCA.

\begin{figure}[!ht]
\centering
\includegraphics[width=1.0\textwidth]{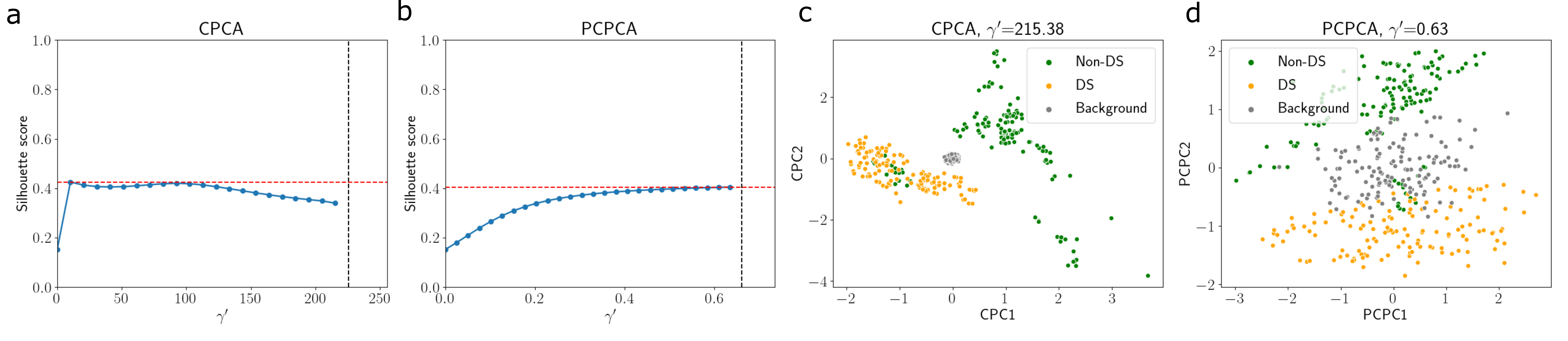}
\caption{\textbf{Comparison of CPCA and PCPCA on the mouse protein expression dataset.} (a) and (b) show the silhouette score of the foreground latent variables at a range of values for $\gamma^\prime$ for each method. The dotted vertical lines shows the first value of $\gamma^\prime$ at which each method ``failed,'' and the red horizontal lines indicate the maximium silhouette score achieved by each method. (c) and (d) show the latent variables for each method at the largest value of $\gamma^\prime$ at which each method ``succeeded.'' Color labels are the same as \autoref{fig:pcpca_mice}. }
\label{fig:mouse_cpca_pcpca_comparison}
\end{figure}

\subsubsection{Single-cell RNA sequencing data}
To test our model in a high-dimensional setting, we fit PCPCA with $d=2$ to a single-cell RNA sequencing (scRNAseq) dataset~\citep{zheng2017massively}. Here, the foreground dataset contains gene expression measurements from bone marrow mononuclear cells (BMMCs) derived from a patient with acute myeloid leukemia (AML) before and after they received a stem-cell transplant ($n=4501$). 
The background dataset contains gene expression measurements of BMMCs from a healthy patient ($m=1985$). We preprocessed the data by log-transforming and subsetting to the 500 most variable genes, in accordance with previous analyses on these data~\citep{zheng2017massively,abid2018exploring}.

Visualizing the two-dimensional latent variables from PCPCA, we found that the model separates the pre- and post-transplant cells effectively at higher values of $\gamma^\prime$, while PPCA ($\gamma^\prime=0$) fails to do so (\autoref{fig:pcpca_scrnaseq}). Furthermore, we measured the silhouette score for these two foreground subgroups in the CPCA and PCPCA reduced-dimension spaces. Similar to our observation with the mouse protein expression dataset, we found that, for PCPCA, the silhouette score monotonically increased with $\gamma^\prime$, while CPCA's performance peaked at lower allowable values of $\gamma^\prime$ (\autoref{fig:scrnaseq_silhouette}). Additionally, the allowable range for $\gamma^\prime$ in CPCA was again much larger than that for PCPCA. The maximum silhouette scores achieved by each method were roughly equivalent (CPCA: $0.20$, PCPCA: $0.23$).
These results imply that PCPCA is effective with high-dimensional data and further demonstrate the advantage of PCPCA's parameterization over CPCA.

\begin{figure}[!ht]
\centering
\includegraphics[width=1.0\textwidth]{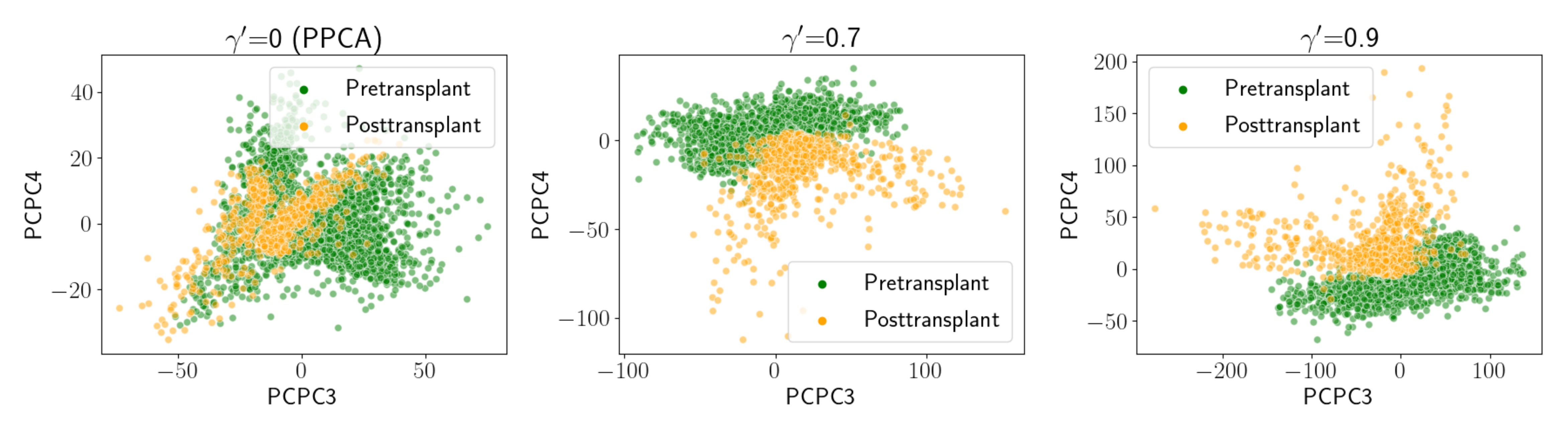}
\caption{\textbf{PCPCA applied to single-cell RNA-seq data.} This dataset contains $n=4501$ foreground samples (cells from AML patient, plotted here) and $m=1985$ background samples (cells from a healthy patient, not plotted here). The foreground cells contain two subgroups: \emph{Pre-transplant} (green) and \emph{Post-transplant} (orange). Plotted here are the two foreground groups projected onto PCPCA's third and fourth components for varying values of $\gamma^\prime$.
}
\label{fig:pcpca_scrnaseq}
\end{figure}

\begin{figure}[h!]
\centering
\includegraphics[width=0.7\textwidth]{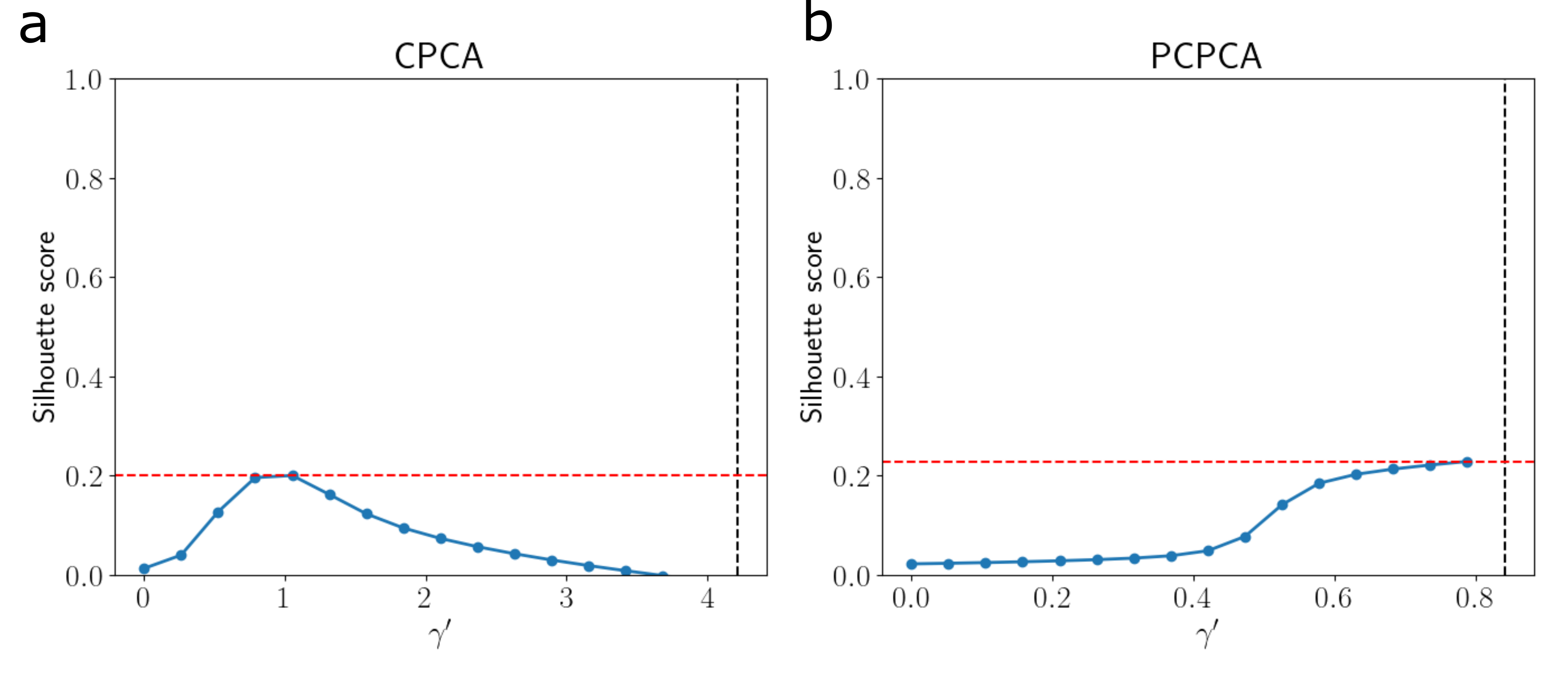}
\caption{\textbf{Silhouette scores for PCPCA and CPCA clusters in scRNA-seq data.} Shown here are the silhouette score across a range of $\gamma^\prime$ for (a) CPCA and (b) PCPCA. The dotted vertical lines shows the first value of $\gamma^\prime$ at which each method ``failed,'' and the red horizontal lines indicate the maximum silhouette score achieved by each method.  }
\label{fig:scrnaseq_silhouette}
\end{figure}

\subsection{Robustness to noise}
An advantage of PCPCA's model-based approach to contrastive learning is its ability to explicitly account for noise in the data. To test this directly, we again fit PCPCA and CPCA to the mouse protein expression dataset, but this time we injected additive, independent Gaussian noise across the features. In particular, we transformed every foreground and background sample $x_i$ and $y_j$ as
\begin{align*}
    \widetilde{x}_i &= x_i + \epsilon_i \\
    \widetilde{y}_j &= y_j + \epsilon_j,
\end{align*}
where $\epsilon_i, \epsilon_j \sim \mathcal{N}(0, \sigma^2 I_D)$. We generated ten datasets for $\sigma^2 \in \{0.5, 1, \dots, 5\}$. We also included the case when $\sigma^2=0$, which is the original dataset with no additional noise.

We fit PCPCA and CPCA on each of these noisy datasets and measured the silhouette score of PCPCA and CPCA with $d=2$. We repeated this experiment $100$ times for each value of $\sigma^2$. We tuned $\gamma$ independently for PCPCA and CPCA for each value of $\sigma^2$ and took the $\gamma$ with the highest silhouette score. We found that, while the performance of both methods declined with more noise, PCPCA showed better performance than CPCA at higher noise levels (\autoref{fig:vary_sigma2}).
This suggests that PCPCA is more robust to noise than CPCA, demonstrating another advantage of our model-based approach.

\begin{figure}[h!]
\centering
\includegraphics[width=0.5\textwidth]{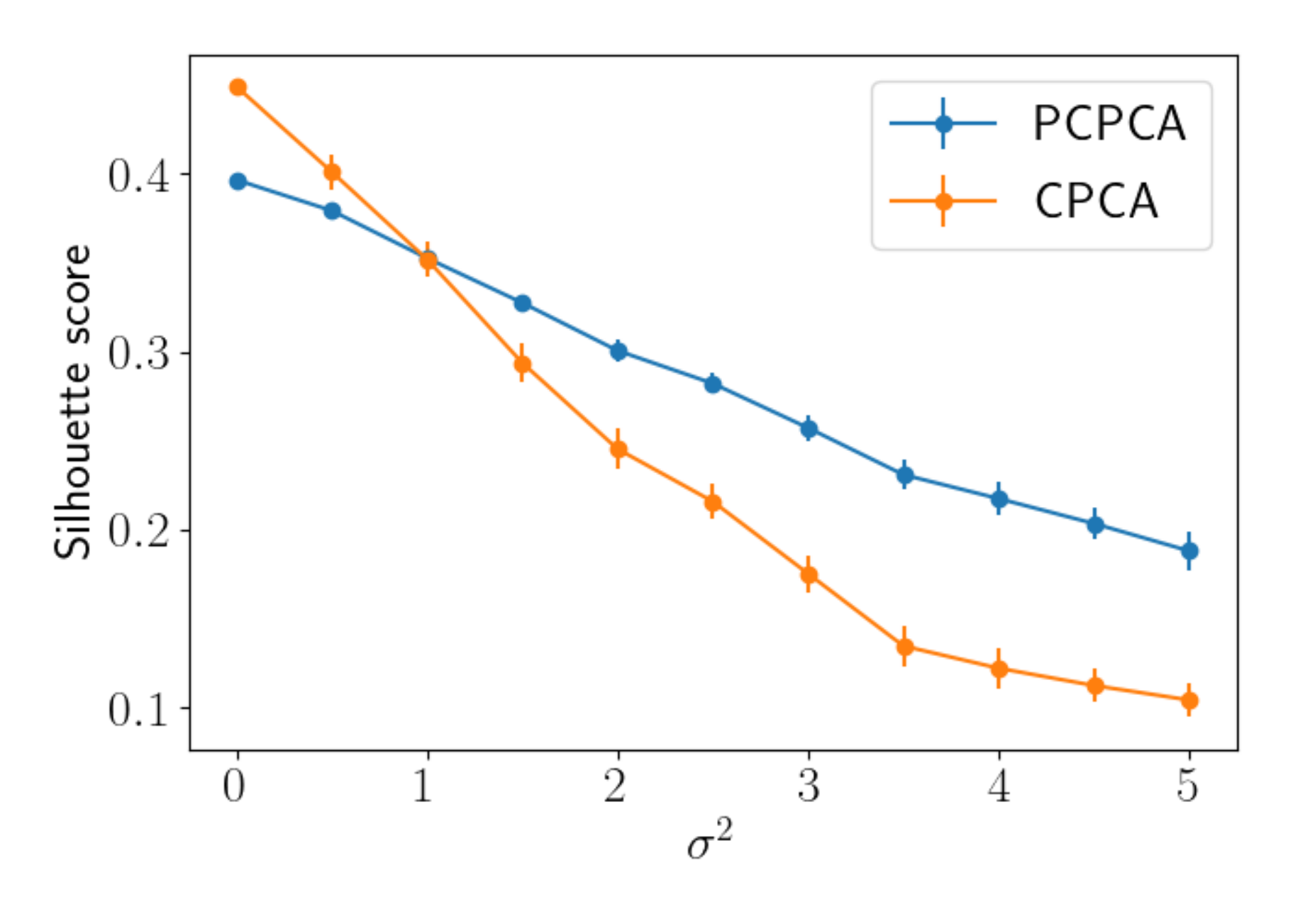}
\caption{\textbf{Model performance with increasing noise.} We injected independent additive Gaussian noise with variance $\sigma^2$ to the mouse protein expression dataset. We then measured the silhouette score of the foreground latent variables. CPCA is shown in orange, and PCPCA is show in blue, both with 95\% confidence interval whiskers.
}
\label{fig:vary_sigma2}
\end{figure}

\subsection{Generating data from the foreground distribution}
Another advantage of PCPCA's model-based approach is the ability to generate data from the foreground data distribution. In CPCA, this is not possible because there is no associated generative model. Note that, in PCPCA, we cannot reasonably generate data from the background distribution because the objective function is a relative likelihood with the goal of minimizing the relative likelihood of the background model. This is not a problem in most settings, as we are typically interested in exploring the variance unique to the foreground data.

To demonstrate PCPCA's ability to generate realistic foreground data, we used the corrupted MNIST dataset~\citep{abid2018exploring}. In this dataset, the foreground samples are MNIST digits (0s and 1s) superimposed onto natural images of grass from ImageNet~\citep{russakovsky2015imagenet}. The background samples are unaltered natural images of grass (\autoref{fig:corrupted_mnist}).
\begin{figure}[h!]
\centering
\includegraphics[width=0.8\textwidth]{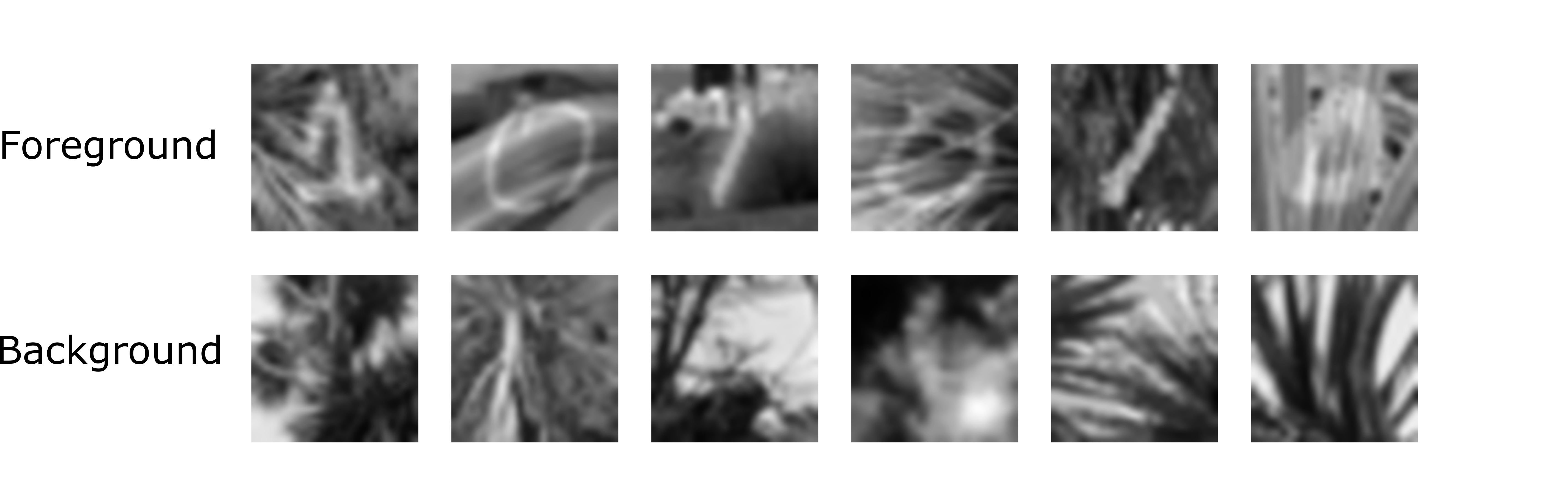}
\caption{\textbf{Examples of the corrupted MNIST dataset.} The top row contains examples from the foreground data, and the bottom row contains examples from the background data. }
\label{fig:corrupted_mnist}
\end{figure}

We fit PCPCA with $d=2$ and $\gamma^\prime=0.8$ to obtain $\widehat{W}_{ML}$. For comparison, we also fit PPCA ($\gamma^\prime=0$). Examining the latent variables, we found that PCPCA showed substantially better clustering of the two MNIST digits than PPCA --- the silhouette score for PCPCA was 0.33, while the score for PPCA was just 0.007 (\autoref{fig:generated_corrupted_mnist}a, c).

To generate new data, we sampled $S=300$ i.i.d. latent variables $z_s \sim \mathcal{N}(0, I_d)$ for $s = 1, \dots, S$, and projected these to the data space to obtain synthetic images. Specifically, each generated image is computed as $\widehat{x}_s = \widehat{W}_{ML} z_s + \mu_x$ where $\mu_x$ is the mean of the foreground data. We found that these samples recovered the variation in the MNIST digits in the foreground data (\autoref{fig:generated_corrupted_mnist}d). In contrast, samples generated from PPCA did not show as much of the digit structure (\autoref{fig:generated_corrupted_mnist}b).
These results suggest that PCPCA can generate realistic data from the foreground distribution, which is useful for exploratory data analysis.
\begin{figure}[h!]
\centering
\includegraphics[width=1.0\textwidth]{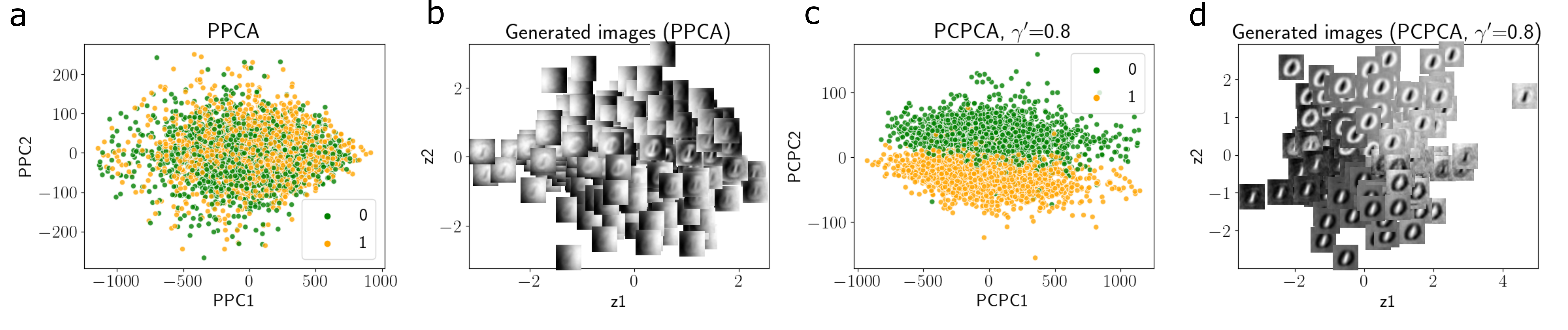}
\caption{\textbf{MNIST digits generated from the PCPCA model.} (a) and (c) show the projected foreground samples for PPCA ($\gamma^\prime=0$) and PCPCA ($\gamma^\prime=0.8)$, respectively. (b) and (d) show new foreground samples generated from the foreground distribution of PPCA and PCPCA, respectively. }
\label{fig:generated_corrupted_mnist}
\end{figure}

Furthermore, using $\widehat{W}_{ML}$ estimated for PPCA and PCPCA fit to the corrupted MNIST data, we computed the log likelihood of a set of held-out samples of MNIST digits without any corruption. 
We found that PCPCA has a higher test likelihood than PPCA on these uncorrupted digits (\autoref{fig:mnist_test_ll}). This suggests that the foreground model for PCPCA more accurately captures the uncorrupted MNIST digits relative to PPCA.

\begin{figure}[h!]
\centering
\includegraphics[width=0.5\textwidth]{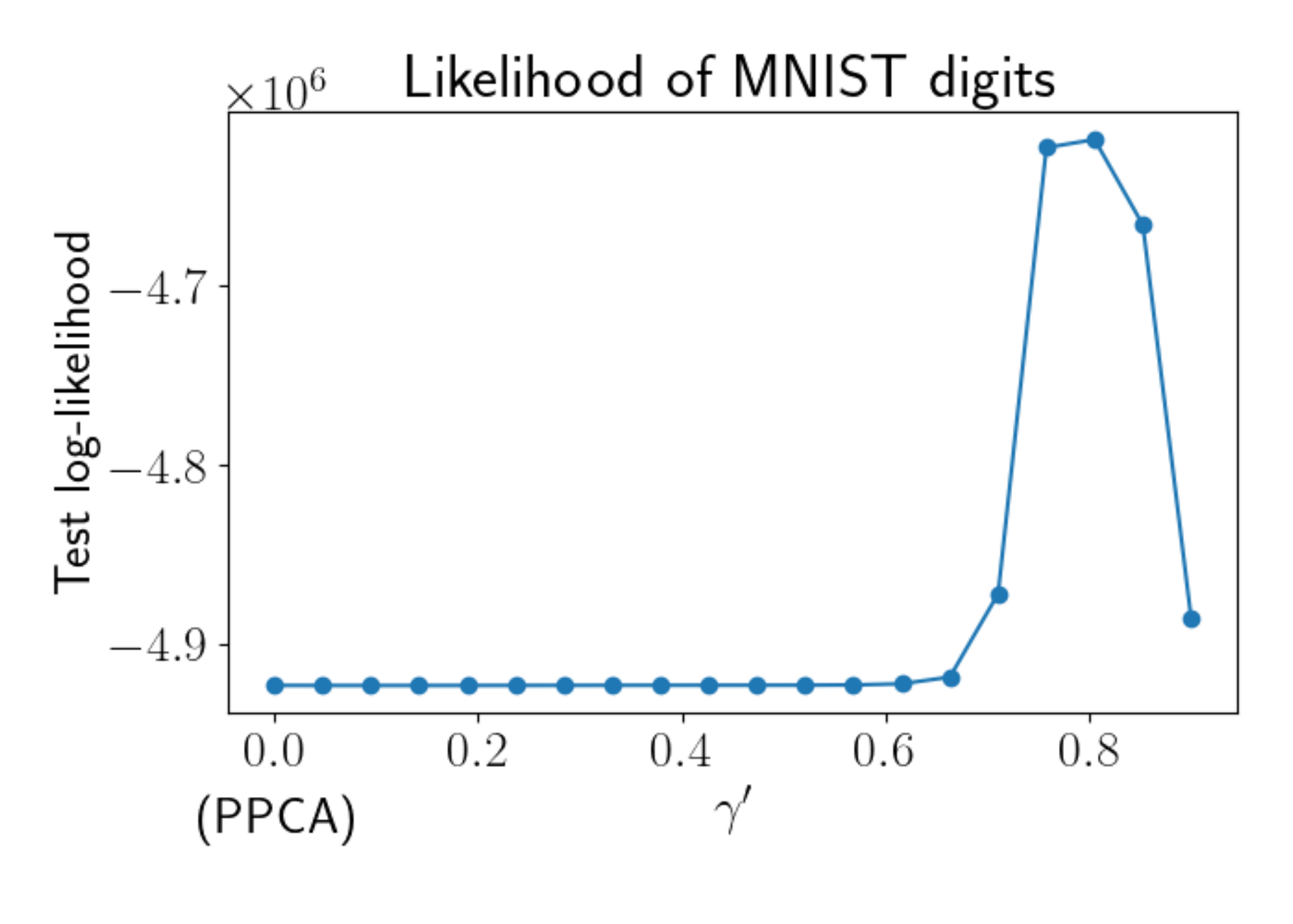}
\caption{\textbf{Log-likelihood of held-out MNIST digits.} Using a PCPCA model fit on the corrupted MNIST dataset, plotted here is the log-likelihood of a held-out set of uncorrupted MNIST digit images. Note that $\gamma^\prime=0$ corresponds to PPCA. }
\label{fig:mnist_test_ll}
\end{figure}

\subsection{Gibbs posterior sampling}
We next sought to numerically evaluate the Gibbs posterior for PCPCA. To estimate the posterior, any sampling-based inference methods can be applied. We use the No U-Turn Sampler~\citep{hoffman2014no} --- which is an extension of Hamiltonian Monte Carlo --- as implemented in the Stan programming language~\citep{carpenter2017stan}. We place uniform priors on $W$ and $\sigma^2$, as required by our theoretical results. 

\subsubsection{Visualizing Gibbs posterior samples}
First, we sought to visualize the posterior for $W$. To do so, we used the same toy dataset as our initial experiments in \autoref{fig:pcpca_toy}. Recall that these samples are generated from a mixture model in which the background distribution is a two-dimensional Gaussian, and the foreground distribution is a mixture of two Gaussians. Specifically,
$$p(x) = \beta \left\{ \pi \mathcal{N}(x | \mu_{f1}, \Sigma_{f1}) + (1 - \pi) \mathcal{N}(x | \mu_{f2}, \Sigma_{f2}) \right\} + (1 - \beta) \mathcal{N}(x | \mu_b, \Sigma_b)$$
where $\beta$ controls the mixture proportion between the background and foreground, and $\pi$ controls the mixture proportion between the foreground subgroups. In this case, we set $\beta=\pi=0.5$, $\Sigma_b = \Sigma_f = \left(\begin{smallmatrix} 4.0 & 2.6 \\ 2.6 & 4.0 \end{smallmatrix}\right)$, $\mu_b = \left(\begin{smallmatrix} 0 \\ 0 \end{smallmatrix}\right)$, $\mu_{f1} = \left(\begin{smallmatrix} -1.5 \\ 1.5 \end{smallmatrix}\right)$, $\mu_{f1} = \left(\begin{smallmatrix} 1.5 \\ -1.5 \end{smallmatrix}\right)$, and $\gamma=0.85$. Using this model, we generated three datasets with increasing sample sizes, respectively containing $30$, $60$, and $300$ samples in each condition. After estimating the Gibbs posterior using each dataset, we drew $100$ samples from the posterior for $W$ (\autoref{fig:figure13}). 

As expected, we found that the posterior increasingly concentrated around the axis separating the foreground subgroups as the sample size increased. With $n=30$, the posterior draws for $W$ were close to uniformly distributed, but with $n=300$, the posterior became tightly concentrated around the desired axis. This suggests that the PCPCA Gibbs posterior is a viable tool for accounting for uncertainty in the context of our loss-based modeling framework. Furthermore, it confirms that the posterior can be estimated using well-known MCMC methods, not requiring any model-specific algorithms.

\begin{figure}[h]
\centering
\includegraphics[width=1.0\textwidth]{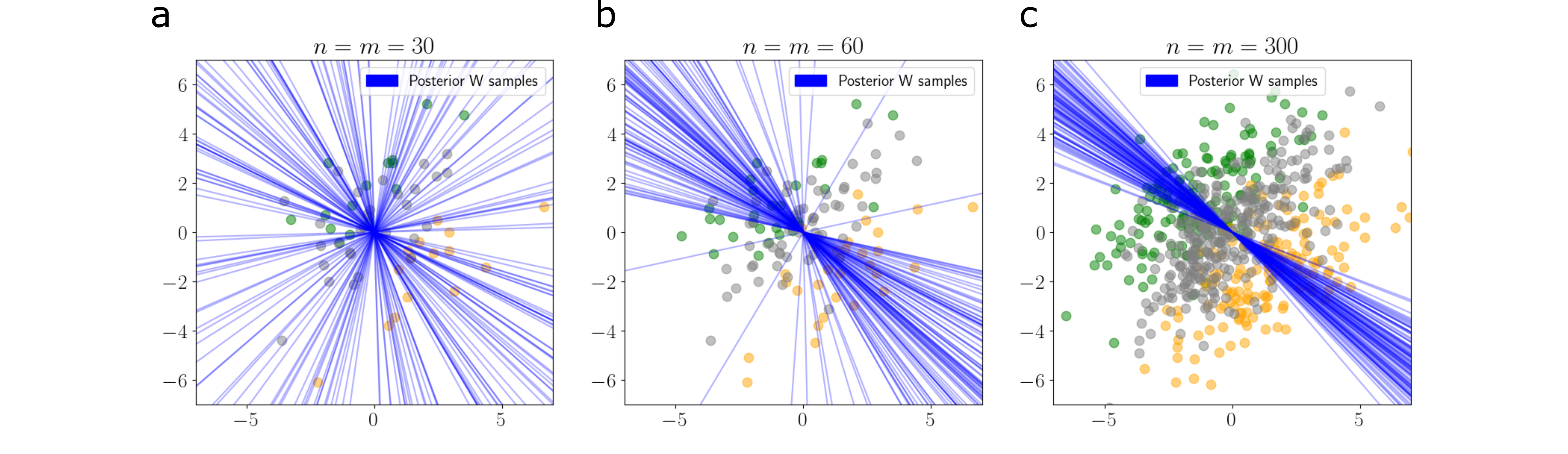}
\caption{\textbf{Sampling from the Gibbs posterior.} Shown here are samples from the Gibbs posterior for $W$ using the toy data. We drew $100$ samples from the posterior for varying sample sizes: (a) 30, (b) 60, (c) 300.}
\label{fig:figure13}
\end{figure}

\subsubsection{Posterior convergence rate}
Next, we sought to validate the posterior convergence rate of $n^{-1/2}$ for the PCPCA Gibbs posterior. To do so, we again simulated data from a mixture of two-dimensional Gaussians. However, in this experiment, we set the foreground to be a single multivariate Gaussian, rather than a mixture of two Gaussians. 

To estimate the convergence rate, we first fit the Gibbs posterior $\Pi_n$ to the simulated data, setting $d=2$ in this case. We then sampled $T=1,000$ parameter values from the posterior, $(W_1, \sigma^2_1), \dots, (W_T, \sigma^2_T) \sim \Pi_n$. We estimated the divergence using each of these samples and the true risk-minimizing parameter values $(W^\star, \sigma^{2\star})$ in \autoref{eq:risk_minizer}. Recall that, in this case, the divergence is $d(\theta, \theta^\star) = (R(\theta) - R(\theta^\star))^{1/2}$, where $\theta = (W, \sigma^2)$ and $R(\cdot)$ is the risk (\autoref{eq:pcpca_risk}). Finally, we computed the fraction of these divergences that exceeded $n^{-1/2}$. Specifically, we computed
\begin{equation*}
\widehat{D} = \frac1T \sum\limits_{t=1}^T I\left(d(\theta_t, \theta^\star) > Cn^{-1/2}\right),
\end{equation*}
where $I$ is the indicator function. We estimated this quantity for $n \in \{50, 200, 300, 400, 500\}$, where $n$ is the total number of samples across conditions. We repeated this five times for each value of $n$.

We found that $\widehat{D}$ fell to zero as $n$ increased, which matches our theoretical result (\autoref{fig:figure14}). This result numerically validates PCPCA's posterior convergence rate of $n^{-1/2}$, which is optimal. It also provides further evidence that the PCPCA Gibbs posterior is a principled tool for performing inference in our framework.

\begin{figure}[h]
\centering
\includegraphics[width=0.5\textwidth]{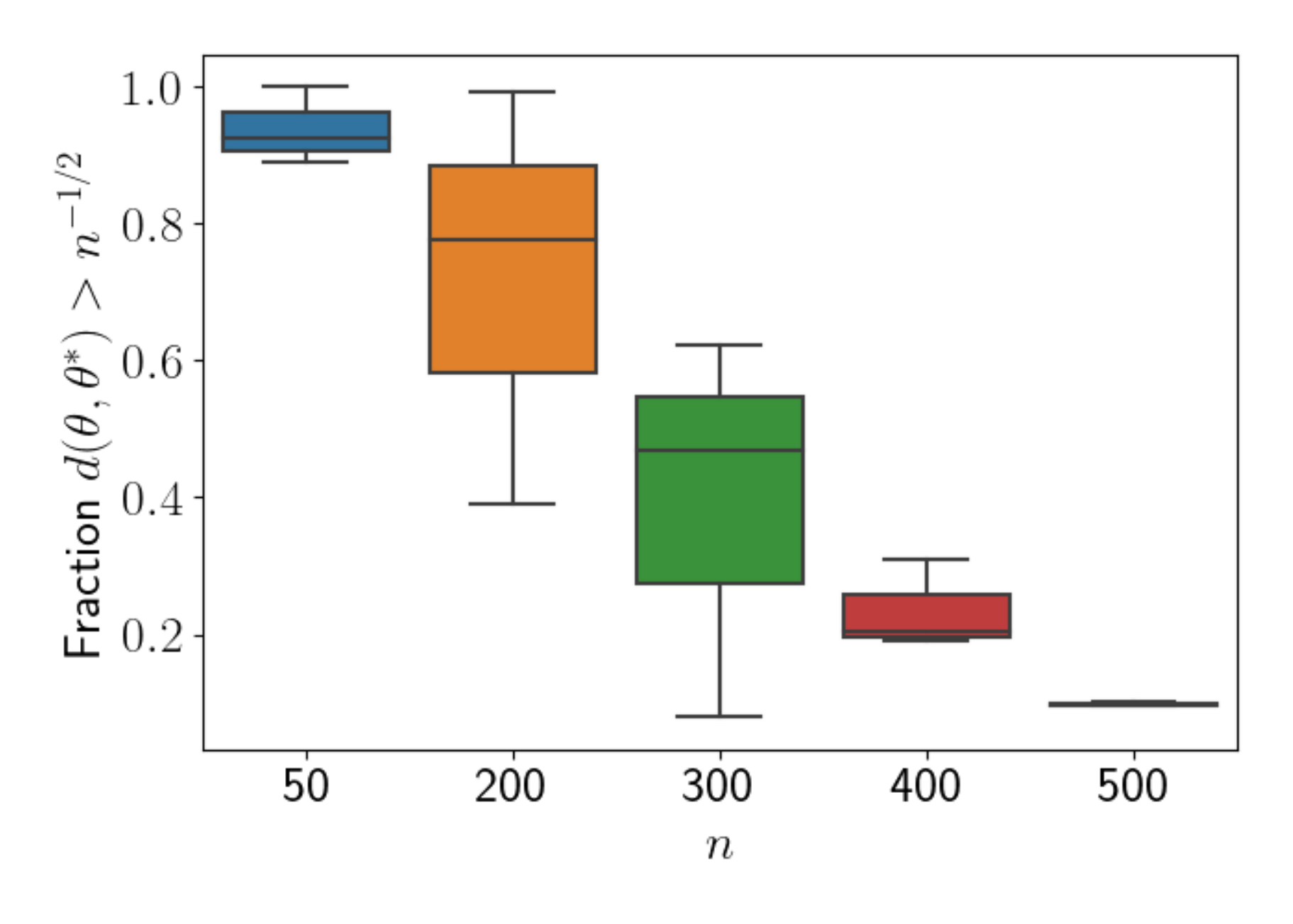}
\caption{\textbf{Empirical validation of the posterior contraction rate for PCPCA Gibbs posterior.} For increasing values of $n$, we randomly drew $1000$ samples for $W$; computed the divergence for each; and computed the fraction of these divergences that exceeded $n^{-1/2}$. The boxes show the results for ten repetitions for each value of $n$.}
\label{fig:figure14}
\end{figure}

\section{Contrastive PCA for missing data}
Another advantage of our probabilistic modeling approach is the ability to handle missing data in a principled way. Missing or incomplete data is extremely common in real-world datasets. Due to its lack of a probabilistic model, CPCA is unable to deal with missing data. In this section, we show how to find maximum likelihood estimates for PCPCA in settings with missing data, and we demonstrate this method through simple experiments.

PPCA can handle missing data, where the MLE relies on an EM algorithm that iteratively reconstructs the missing matrix elements from the PCs, and re-estimates the PCs from the expected complete matrix~\citep{tipping1999probabilistic, roweis1998algorithms}. However, in PCPCA, the target function \eqref{eqn:PCPCA} to be maximized is not a likelihood, so we cannot apply the EM algorithm. Instead, we propose a data augmentation method by introducing a indicator matrix representing the location of missing elements to obtain closed-form gradients of the objective so that we can make use of existing gradient-based optimization algorithms, such as gradient descent. We first present the details of our approach and then demonstrate its performance through experiments.

\subsection{Gradient descent with missing data}
Assume some elements of both the background and foreground matrices are missing. Let $x=[x^o,x^u]^\top$, where $x^o$ is the sub-vector of observed features and $x^u$ unobserved, and $y=[y^o,y^u]$ with the same partition. Consider the missing-at-random setting, where the MLE of the complete data is the same as the MLE of the non-missing data only. Let $x_i^o$ be the observed subvector of $x_i$ with length $D_i$, where the locations observed are $i_1,\cdots,i_{D_i}$. Then, we introduce a indicator matrix $L_i$ with dimension $D_i\times D$ such that $x_i^o=L_ix_i$:
$$(L_i)_{kl} = \begin{cases} 1 & l=i_k\\
0 & \text{otherwise}.
\end{cases}$$ 
Similarly, let $y_j^o$ be the observed subvector of $y_j$ with length $E_j$, and define $M_j\in\RR^{E_j\times D}$ as before such that $y_j^o=M_j y_j$.
Observe that $$x_i^o\sim N(0,A_i),~~y_j^o\sim N(0,B_j),~~A_i\coloneqq L_i(WW^\top+\sigma^2\Id_D)L_i^\top,~~B_j\coloneqq M_j(WW^\top+\sigma^2\Id_D)M_j^\top.$$
As a result, the objective function of the observed data is
\begin{align*}
&l(W,\sigma^2)=l(X^o|W,\sigma^2)-\gamma l(Y^o|W,\sigma^2)\\
&= -\frac{1}{2}\sum_{i=1}^n\left(D_i\log (2\pi)+\log \det(A_i)+\tr( A_i^{-1}x_i^o{x_i^o}^\top)\right)+\frac{\gamma}{2}\sum_{j=1}^m\left(E_i\log (2\pi)+\log \det(B_j)+\tr( B_j^{-1}y_j^o{y_j^o}^\top)\right).
\end{align*}
Then we take the derivative w.r.t. to $W$:
\begin{align*}
    \frac{\partial l}{\partial W}& =-\left\{\sum_{i=1}^n L_i^\top A_i^{-1}\left(\Id_{D_i}-x_i^o{x_i^o}^\top A_i^{-1}\right)L_i-\gamma \sum_{j=1}^m M_j^\top B_j^{-1}\left(\Id_{E_j}-y_j^o{y_j^o}^\top B_j^{-1}\right)M_j\right\}W.
\end{align*}
Similarly, the derivative w.r.t. $\sigma^2$ is
\begin{align*}
    \frac{\partial l}{\partial \sigma^2}& =-\frac{1}{2} \sum_{i=1}^n\left(\tr(A_i^{-1}L_iL_i^\top)-\tr(A_i^{-1}x_i^o{x_i^o}^\top A_i^{-1}L_iL_i^\top)\right)\\
    &~~~+\frac{\gamma}{2} \sum_{j=1}^m\left(\tr(B_j^{-1}M_jM_j^\top)-\tr(B_j^{-1}y_j^o{y_j^o}^\top B_j^{-1}M_jM_j^\top)\right).
\end{align*}
We can then use iterative optimization algorithms, such as gradient descent, to find $\widehat{W}_{ML}$ and $\widehat{\sigma}^2_{ML}$.

\subsection{Imputing missing data}
After finding $\widehat{W}_{ML}$ and $\widehat{\sigma}^2_{ML}$ as above, the unobserved foreground values can be imputed. Let $P_i$ be an indicator matrix with dimension $U_i\times D$, where $U_i=D-D_i$, such that $x_i^u=P_i x_i$. Further, define
\begin{equation*}
    C_i\coloneqq P_i(\widehat{W}_{ML}\widehat{W}_{ML}^\top + \widehat{\sigma}_{ML}^2 \Id_D)P_i^\top,~~F_i\coloneqq P_i(\widehat{W}_{ML}\widehat{W}_{ML}^\top + \widehat{\sigma}_{ML}^2 \Id_D)L_i^\top.
\end{equation*}
Continuing to assume mean-centered data, observe that
\begin{equation*}
    x_i^u | x_i^o \sim N(F_i A_i^{-1} x_i^o, C_i - F_i A_i^{-1} F_i^\top).
\end{equation*}
The unobserved values can then be imputed using the conditional mean $\widehat{x}_i^u = F_i A_i^{-1} x_i^o$.

\subsection{Experiments with missing data}
\subsubsection{Simulated data}
To test PCPCA in the presence of missing data, we first fit the model to a synthetic dataset. 

To construct the dataset, we generated foreground and background data from separate PPCA models in order to give them separate covariance structures. In particular, for $i \in [n]$ and $j \in [m]$, we sampled $x_i \sim \mathcal{N}(W^\text{f} z_i, \sigma^2 \Id_D)$ and $y_j \sim \mathcal{N}(W^\text{b} z_j, \sigma^2 \Id_D)$ where $z_i, z_j \sim \mathcal{N}(0, \Id_d)$. In our experiments, we set $n=m=100$, $D=10$, $d=2$, and $\sigma^2=1$. We sampled the elements of $W^\text{f}$ and $W^\text{b}$ independently at random from a standard Gaussian.

To test the performance of PCPCA, we randomly removed elements of these two matrices with probability $p$, simulating a missing-at-random scenario. In our experiments, we used $p \in \{0, 0.1, 0.2, \dots, 0.7\}$. After removing the randomly chosen values, we fit PCPCA on the partially observed dataset using the gradients derived in the previous section, along with the Adam optimizer for additional stability~\citep{kingma2014adam}. Using the fitted model, we computed the log likelihood of a held-out dataset of foreground data. For comparison, we also fit PPCA on the pooled data $X \cup Y$ and reported the log likelihood of the held-out foreground dataset.

We found that PCPCA showed relatively steady test log likelihood for $p\leq0.3$ (\autoref{fig:pcpca_missingdata_sim}a). For higher levels of missing data, PCPCA showed a steady decline in test log likelihood. In contrast, PPCA showed a substantially lower test log-likelihood than PCPCA at all values of $p$. Meanwhile, CPCA and PCA do not allow for settings where $p>0$.

For each of these partially-observed datasets, we also imputed the missing values in the foreground data and computed the reconstruction error. As before, let $x_i^u \in \mathbb{R}^{U_i}$ be the true values for the unobserved portion of $x_i$, and let $\widehat{x}_i^u$ be the PCPCA reconstruction of these values. We computed the mean-squared error of these reconstructions:
\begin{equation*}
    \text{MSE} = \frac1n \sum\limits_{i=1}^n \frac{1}{U_i} \|x_i^u - \widehat{x}_i^u\|_2^2.
\end{equation*}
We found that PCPCA achieved low reconstruction error when few values were missing, and the error increased steadily for higher fractions of unobserved values (\autoref{fig:pcpca_missingdata_sim}b). In all cases, PCPCA performed better than PPCA.

These results suggest that PCPCA is robust in the presence of missing data, even when a relatively large fraction of the data is missing at random.

\begin{figure}[h!]
\centering
\includegraphics[width=1.0\textwidth]{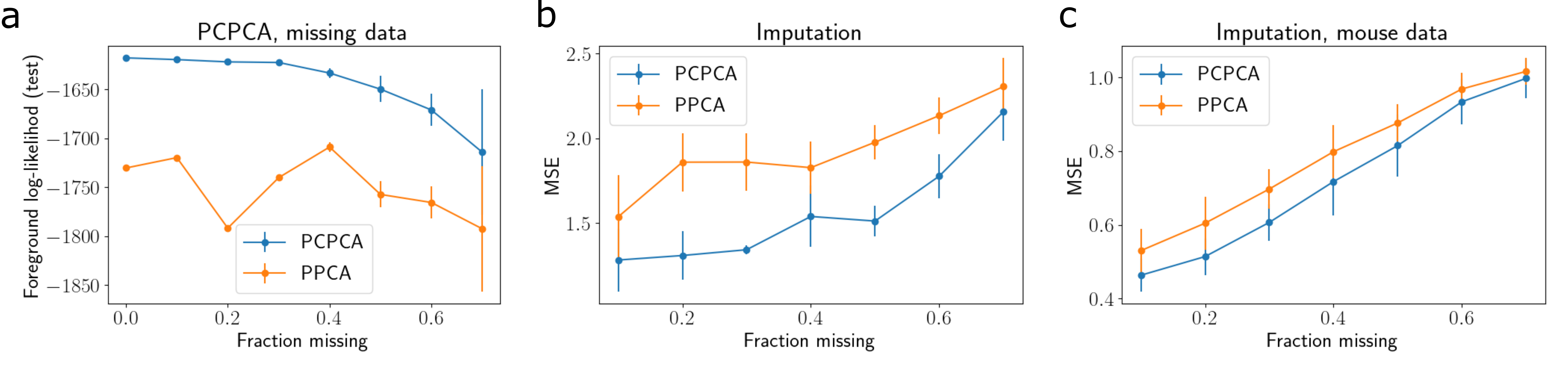}
\caption{\textbf{PCPCA and PPCA with missing data.} In (a) and (b), we used a simulated dataset in which $n=m=100$, $D=10$, and $X$ and $Y$ were generated with separate PPCA models. We fit PCPCA on $X$ and $Y$ and PPCA on the concatenation of $X$ and $Y$. Values were dropped with increasing probability. (a) shows the test log-likelihood was computed for a held-out foreground dataset. (b) shows the mean-squared error of the imputed values for the training data from the PCPCA model. (c) shows the imputation error for the mouse protein expression data. }
\label{fig:pcpca_missingdata_sim}
\end{figure}

\subsubsection{Mouse protein expression data}
To further validate PCPCA's ability to handle missing data, we applied the model to the mouse protein expression dataset. Here, we randomly removed elements from the foreground and background data with probability $p$, where $p \in \{0, 0.2, 0.6, 0.9\}$. We fit the model to each partially-observed dataset using gradient descent and the Adam optimizer to obtain $\widehat{W}_{ML}$. We set $d=2$ and $\gamma=0.4$ for all runs based on previous experiments. Finally, we projected the fully-observed dataset to the latent space and computed the silhouette score.

We observed that the two subgroups of mice were preserved even with a large fraction of the data masked (\autoref{fig:pcpca_missingdata_mouse}). The silhouette scores confirmed this, staying steady for $p<0.6$. Furthermore, we imputed the missing foreground values using the PCPCA model, and we found that the model's reconstructions consistently showed lower error than PPCA (\autoref{fig:pcpca_missingdata_sim}c).

These results imply that PCPCA could be used in many real-world settings in which datasets are only partially observed.

\begin{figure}[h]
\centering
\includegraphics[width=1.0\textwidth]{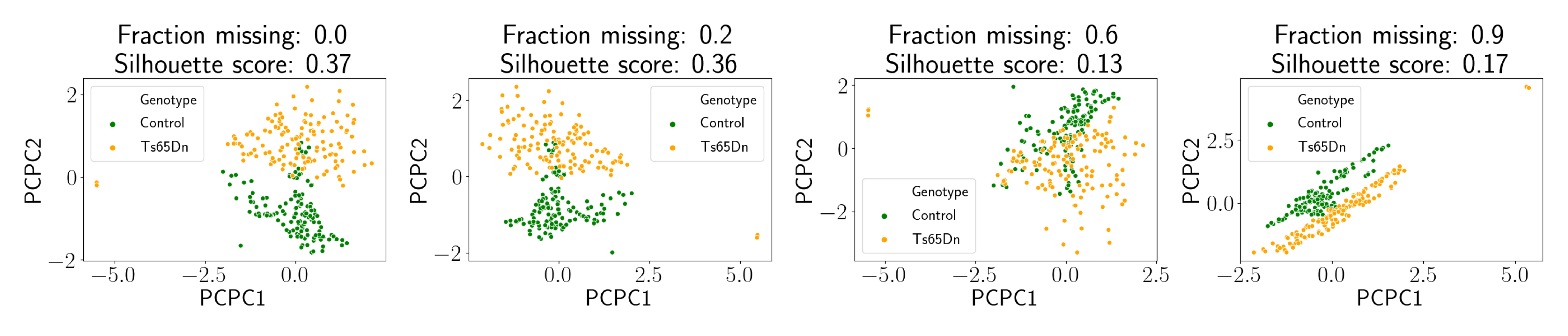}
\caption{\textbf{PCPCA on mouse protein expression data with missing data.} We randomly removed values from the data and fit PCPCA. Shown here are the latent projections of the fully-observed foreground dataset for increasing missing probabilities. }
\label{fig:pcpca_missingdata_mouse}
\end{figure}

\section{Discussion}
In this paper, we presented a probabilistic model, PCPCA, for learning the contrastive dimensions between a foreground dataset and a background dataset. We derived conditions for the tuning parameter $\gamma$ under which our model and a previous method, CPCA, are well-defined. In contrast to CPCA, our model-based approach allows for uncertainty quantification, is robust to noise, includes the ability to sample from the fitted model, and is able to accept and impute missing data. We developed a generalized Bayesian framework that allows for principled inference in our model, despite it not having a well-defined likelihood. To find the loss-minimizer in the presence of missing data, we derived a gradient descent algorithm that only relies on the observed data. We demonstrated the utility of PCPCA in several applications using protein expression, gene expression, and image data. We found that PCPCA outperformed PPCA and CPCA in capturing subgroup structure and in its robustness to noise and missing data.

Several future directions remain to be explored. First, more general inference procedures could be developed for the PCPCA relative likelihood objective function. Likelihood ratios have been well-studied for ratios comparing two sets of parameters under a single shared dataset~\citep{anderson1962introduction}. However, there has been little work studying relative likelihoods that compare one set of shared parameters under two datasets. Typical inference procedures --- such as expectation maximization (EM) -- cannot be used in this contrastive setting because the objective function is not a well-defined likelihood. Future work will benefit from adapting well-known methods, such as EM, for estimation, inference, and optimization to the novel relative likelihood objective presented in this work.

Second, more sophisticated optimization procedures could be used for performing gradient descent in the presence of missing data. In both PPCA and PCPCA, the gradient of the likelihood w.r.t. the noise variance $\sigma^2$ shows a sharp increase as $\sigma^2 \to 0$. Modern optimization techniques could be used to further stabilize the gradient in this regime.

Finally, the PCPCA model itself could be extended in several ways. Future versions could incorporate various data likelihoods in order to capture non-Gaussian data. Other extensions might allow multiple foreground datasets, possibly allowing structured relationships between those foreground matrices. Furthermore, a non-linear version of PCPCA could be considered.

\newpage
\bibliographystyle{apalike}
\bibliography{ref.bib}

\section{Appendix}
\subsection{Data availability}
All data are available via their respective papers. Preprocessing scripts are included in the code repository: \url{https://github.com/andrewcharlesjones/pcpca}.

\subsection{Proof of Theorem \ref{thm:CPCAloss}}\label{pf:CPCAloss}
\begin{proof}
We simplify the geometric objective function in \eqref{eqn:CPCAl2} until it matches \eqref{eqn:CPCAvar}, the statistical objective function. Let $v_1$ be the eigenvector of $C$ corresponding to the largest eigenvalue, then
\begin{align*}
  v_1 &=\underset{v^\top v=1}{\argmin}~\frac{1}{n} \sum_{i=1}^n \|x_i-vv^\top x_i\|^2-\gamma \frac{1}{m} \sum_{j=1}^m \|y_j-vv^\top y_j\|^2\\
  & = \underset{v^\top v=1}{\argmin}~\frac{1}{n} \sum_{i=1}^n \left(x_i^\top x_i-2x_i^\top vv^\top x_i+x_ivv^\top vv^\top x\right)-\gamma \frac{1}{m} \sum_{j=1}^m \left(y_j^\top y_j-y_j^\top vv^\top y_j+y_j^\top vv^\top vv^\top y_j\right)\\
  & = \underset{v^\top v=1}{\argmin}~\frac{1}{n} \sum_{i=1}^n \left(-2x_i^\top vv^\top x_i+x_ivv^\top x\right)-\gamma \frac{1}{m} \sum_{j=1}^m \left(-y_j^\top vv^\top y_j+y_j^\top vv^\top y_j\right)\\
  & = \underset{v^\top v=1}{\argmax}~\frac{1}{n} \sum_{i=1}^n x_i^\top vv^\top x_i-\gamma \frac{1}{m} \sum_{j=1}^m y_j^\top vv^\top y_j\\
  & = \underset{v^\top v=1}{\argmax}~v^\top \frac{1}{n}\sum_{i=1}^n x_ix_i^\top v-\gamma v^\top \frac{1}{m}\sum_{j=1}^m y_jy_j^\top v\\
  &=\underset{v^\top v=1}{\argmax}~v^\top C_xv-\gamma v^\top C_Yv.
\end{align*}
\end{proof}

\subsection{Proof of Lemma \ref{lem:PD}}\label{pf:PD}
\begin{proof}
Let $\lambda_1\geq \cdots\geq \lambda_D$ and $\rho_1\geq \cdots\geq \rho_D$ be the eigenvalues of $C_X$ and $C_Y$ in descending order. Assume $\gamma<\lambda_D/\rho_1$; then for any unit vector $v\in\RR^D$ with $\|v\|=1$, 
\begin{align*}
    v^\top C v &= v^\top C_Xv-\gamma v^\top C_Y v \\
    & \geq \lambda_D v^\top v-\gamma \rho_1 v^\top v\\
    &=\lambda_D-\gamma \rho_1>0,
\end{align*}
so $C$ is positive definite. 

Assume $\gamma \geq \lambda_D/\rho_1$ with corresponding eigenvectors $u_D$ and $v_1$ where $u_D=v_1$. Let $u=u_D=v_1$ with $\|u\|=1$, then
\begin{align*}
    u^\top Cu& = u^\top C_Xv-\gamma u^\top C_Yu\\
    &=  \lambda_D - \gamma \rho_1\leq 0.
\end{align*}
So $C$ is not positive definite.

\end{proof}
Note that the second half of the proof is the worst case, where the last eigenvector of $C_X$ matches the first eigenvector of $C_Y$. In order to make $C$ positive definite (PD), the strong condition is necessary. However, in practice, much larger values of $\gamma$ are sometimes allowed such that $C$ is still PD.

\subsection{Proof of Theorem \ref{thm:Weyl}}\label{pf:Weyl}
\begin{proof}
Recall that the eigenvalues of $C$, $C_X$ and $C_Y$ are $\lambda_1\geq \cdots\lambda_D\geq0$, $\mu_1\geq \cdots\mu_D\geq0$ and $\rho_1\geq\cdots\geq\rho_D\geq 0$. Then the eigenvalues of $-\gamma C_Y$ are $-\gamma \rho_D\geq\cdots\geq-\gamma \rho_1$. Since $C=C_X-\gamma C_Y=C_X+(-\gamma C_Y)$, by Weyl's inequalities, for any $j+k\geq D+d$, 
$$\lambda_d\geq \mu_j-\gamma \rho_{D-k+1}>0.
$$
Similar to the proof of Lemma \ref{lem:PD}, if the condition is violated, there exists a $C$ such that the first $d$ eigenvalues are negative.
\end{proof}

\subsection{Proof of Corollary \ref{cly:loss}}\label{pf:loss}
\begin{proof}
By the same proof as the proof for the PCA loss, the CPCA loss is $\sum_{i=d+1}^D\lambda_i$. For a fixed $d$, we first show that increasing $\gamma$ will result in a smaller $\lambda_i$. Let $\gamma_1 < \gamma_2$,  $C_1=C_X-\gamma_1 C_Y$ and $C_2 = C_Y-\gamma_2 C_Y$. Then $C_2-C_1 =(\gamma_1-\gamma_2)C_Y$ is positive definite, hence the eigenvalues of $C_2$ are greater than those of $C_1$. This implies that increasing $\gamma$ will decrease the loss.

Then, for a fixed $\gamma<\max\left\{\frac{\eig_{d+1}(C_X)}{\eig_{1}(C_Y)},\frac{\eig_{d+1}(C_X)}{\eig_{2}(C_Y)},\cdots,\frac{\eig_D(C_X)}{\eig_{D-d}(C_Y)}\right\}$, Theorem \ref{thm:CPCAloss} implies $\lambda_{d+1}>0$, so raising $d$ to $d+1$ results in a smaller tail sum of the eigenvalues, hence a smaller loss. 
\end{proof}

\subsection{Proof of Theorem \ref{thm:PCPCA}}\label{pf:PCPCA}
First, we find the maximizer of $W$ given $\sigma^2$. Recall the marginals: $x,y\sim N(0,WW^\top+\sigma^2 \Id_D)$, and denote $A = WW^\top+\sigma^2\Id_D$. Then, taking the log of the objective function, we have
$$l(W,\sigma^2|X,Y,\sigma^2) = -\frac{n}{2}\left(D\ln(2\pi)+\ln|A|+\tr(A^{-1}C_X)\right)+\frac{\gamma m}{2}\left(D\ln(2\pi)+\ln|A|+\tr(A^{-1}C_Y)\right).$$
We drop all constants and the log likelihood becomes
$$l(W,\sigma^2)=-\frac{n-\gamma m}{2}\ln|A|-\frac{1}{2}\tr(A^{-1}(nC_X-\gamma m C_Y))=-\frac{n-\gamma m}{2}\ln|A|-\frac{1}{2}\tr(A^{-1}C).$$
Denote $C=nC_X-\gamma m C_Y$ where $C_X=\frac{1}{n}\sum_{i=1}^n x_ix_i^\top$ and $C_Y=\frac{1}{m}\sum_{j=1}^m y_jy_j^\top$. Then, we take the derivative of $l$:
$$\frac{\partial l}{\partial W}=-\frac{n-\gamma m}{2} 2A^{-1}W+\frac{1}{2}2A^{-1}CA^{-1}W.$$
Letting $\frac{\partial l}{\partial W}=0$, we have
$$A^{-1}\frac{\sum_{i=1}^n x_ix_i^\top-\gamma \sum_{j=1}^my_jy_j^\top}{n-\gamma m}A^{-1}W=A^{-1}W,$$
that is, $\frac{1}{n-\gamma m}CA^{-1}W=W$. $C=(n-\gamma m)A$ solves this equation, that is,
$$WW^\top = \frac{1}{n-\gamma m}C-\sigma^2\Id_D.$$
Assume $C=U\Lambda U^\top$, then
$$\widehat{W}_{ML}=U_d\left(\frac{\Lambda_d}{n-\gamma m}-\sigma^2\Id_d\right)^{1/2}R,$$
where $U_d$ consists of the first $d$ eigenvectors of $C$, $\Lambda_d$ contains the corresponding eigenvalues, and $R$ is any rotation matrix.  

Next, we consider $\widehat{\sigma}^2_{ML}$. Plugging $\widehat{W}_{ML}$ into the objective, we have
\begin{align*}
l(\sigma^2|X,Y,\widehat{W}_{ML})&=-\frac{n-\gamma m}{2}\ln|A|-\frac{1}{2}\tr(A^{-1}C)\\
& = -\frac{n-\gamma m}{2}\left(\sum_{i=1}^d\ln\frac{\lambda_i}{n-\gamma m}+(D-d)\ln\sigma^2\right)-\frac{1}{2}\left(d(n-\gamma m) +\frac{1}{\sigma^2}\sum_{j=d+1}^D\lambda_j\right).\\
& = -\frac{n-\gamma m}{2}\left(\sum_{i=1}^d\ln\widetilde{\lambda_i}+(D-d)\ln\sigma^2+d +\frac{1}{\sigma^2}\sum_{j=d+1}^D\widetilde{\lambda_j}\right),
\end{align*}
where $\widetilde{\lambda_i}=\frac{\lambda_j}{n-\gamma m}$. The derivative of $l$ is:
\begin{align*}
    \frac{\partial l}{\partial\sigma^2}& = -\frac{n-\gamma m}{2}\left(\frac{D-d}{\sigma^2} -\frac{1}{\sigma^4}\sum_{j=d+1}^D\widetilde{\lambda_j}\right).
\end{align*}
Letting $\frac{\partial l}{\partial\sigma^2}=0$, we have
$\frac{D-d}{\sigma^2}=\frac{1}{\sigma^4}\sum_{j=d+1}^D\widetilde{\lambda_j},$
so the MLE of $\sigma^2$ is given by
$$\widehat{\sigma}^2_{ML}=\frac{1}{D-d}\sum_{i=d+1}^D\widetilde{\lambda_i}=\frac{1}{(D-d)(n-\gamma m)}\sum_{i=d+1}^D \lambda_i.$$

We next connect PCPCA to CPCA and PPCA. When $\gamma=0$, the objective function \eqref{eqn:PCPCA} is the same as the objective of PPCA, so the MLEs are also the same. 
Alternatively, when $\sigma^2\to 0$, 
$$\widehat{W}_{ML}=U_d\Lambda^{1/2}R$$ is exactly the solution of CPCA with the new parameterization, which corresponds to eigenvectors of $C=\sum_{i}x_ix_i^\top-\gamma \sum_{j=1}^m y_jy_j^\top$, and is equivalent to the CPCA proposed by \cite{abid2018exploring} with $\gamma'=\gamma m/n$. In this sense, our parameterization is more natural since it corresponds to the likelihood function.

\subsection{Proof of Theorem \ref{thm:CPCArate}}
\label{proof:CPCArate}
First recall the following Lemma.
\begin{definition}
The the loss function $l$ is said to be of sub-exponential type if there exists $\overline{w}, K, r>0$ such that for any $w\in(0,\overline{w})$, 
\begin{equation}\label{eqn:priorcond}
d(\theta;\theta^*)>\delta\Longrightarrow \EE_Pe^{-w(l_\theta-l_{\theta^*})}\leq e^{-Kw\delta^r}. 
\end{equation}
\end{definition}

Let $m(\theta,\theta^*)=\EE_P(l_\theta-l_{\theta^*})=R(\theta)-R(\theta^*)$ and $v(\theta,\theta^*) = \EE_P \left(l_\theta-l_{\theta^*}-m(\theta,\theta^*)\right)^2$. 
\begin{lemma}[\cite{syring2020gibbs}] \label{lem:syring2020gibbs}
Assume $\varepsilon_n\to0$ and $n\varepsilon_n^r\to\infty$ for $r>0$, the prior satisfies
\[\log\Pi(\{\theta:m(\theta,\theta^*)\vee v(\theta,\theta^*)\leq \varepsilon_n^r\})\gtrsim -Mn\varepsilon_n^r\]
and the loss function is of sub-exponential type, then the Gibbs posterior distribution has asymptotic concentration rate $\varepsilon_n$ for all large enough constants $M>0$.
\end{lemma}

By the definition of the divergence: $d(\theta;\theta^*)=(R(\theta)-R(\theta^*))^{1/2}$, 
we know that
$$d(\theta;\theta^*)>\delta\Longrightarrow \EE_P e^{-w(l_\theta-l_{\theta^*})}\leq e^{-w (R(\theta)-R(\theta^*))}\leq e^{-w\delta^2},$$
so the loss is of sub-exponential type.
Then it suffices to check that the prior satisfies the conditions in Lemma \autoref{lem:syring2020gibbs}. First we calculate $m$,
\begin{equation}\label{eqn:CPCA_m}
    m(\theta,\theta^*) = R(\theta)-R(\theta^*) = \tr\left((V^*V^{*\top}-VV^\top)C)\right).
\end{equation}
Recall that $v(\theta,\theta^*) = \EE_P\left(l_\theta- l_{\theta^*}\right)^2-m(\theta,\theta^*)^2$. We show the following two lemmas to calculate $v$.

\begin{lemma}
$X\sim N(0,C)$, then $\EE[X^\top X X^\top X] = \tr(C)^2+2\tr(C^2).$
\end{lemma}
\begin{proof}
Since $\EE [X^\top X X^\top X]=\EE \left(\sum_{i,j} X_i^2X_j^2\right)=\sum_{i,j} \EE(X_i^2X_j^2)$, we start with $\EE (X_i^2X_j^2)$. Observe that $(X_i,X_j)\sim N\left(\begin{bmatrix} 0\\ 0
\end{bmatrix}, \begin{bmatrix} C_{ii} & C_{ij}\\ C_{ij} & C_{jj}
\end{bmatrix}\right)$, so $\EE(X_i^2X_j^2) = C_{ii}C_{jj}+2C_{ij}^2$. Then we have
\begin{align*}
    &\EE [X^\top X X^\top X]=\EE \left(\sum_{i,j} X_i^2X_j^2\right)=\sum_{i,j} \EE(X_i^2X_j^2)\\
    & = \sum_{i,j}\left(C_{ii}C_{jj}+2C_{ij}^2\right)= \tr(C)^2+2\tr(CC^\top).
\end{align*}
\end{proof}

\begin{lemma}
$X\sim N(0,C)$, and $A$ is symmetric, then $\EE[X^\top X X^\top AX] = \tr(C)\tr(AC)+2\tr(AC^2)$.
\end{lemma}
\begin{proof}
Let $Y=A^{1/2}X$, so $Y\sim N(0,A^{1/2}CA^{1/2})$.
Since $\EE [X^\top X X^\top AX]=\EE \left(\sum_{i,j} X_i^2Y_j^2\right)=\sum_{i,j} \EE(X_i^2Y_j^2)$, we start with $\EE (X_i^2Y_j^2)$. Observe that $$(X_i,Y_j)\sim N\left(\begin{bmatrix} 0\\ 0
\end{bmatrix}, \begin{bmatrix} C_{ii} & (CA^{1/2})_{ij}\\ (CA^{1/2})_{ij} & (A^ {1/2}CA^{1/2})_{jj}
\end{bmatrix}\right),$$ so $\EE(X_i^2Y_j^2) = C_{ii}(A^ {1/2}CA^{1/2})_{jj}+2C_{ij}^2$. Then we have
\begin{align*}
    &\EE [X^\top X X^\top AX]=\EE \left(\sum_{i,j} X_i^2Y_j^2\right)=\sum_{i,j} \EE(X_i^2Y_j^2)\\
    & = \sum_{i,j}\left(C_{ii}(A^ {1/2}CA^{1/2})_{jj}+2(CA^{1/2})_{ij}^2\right)= \tr(C)\tr(A^ {1/2}CA^{1/2})+2\tr(A^{1/2}CC^\top A^{1/2})\\
    & = \tr(C)\tr(AC)+2\tr(AC^2).
\end{align*}
\end{proof}
Let $\Delta=V^*V^{*\top}-VV^\top$, then observe that

\begin{align*}
  &\EE_p\left[l_\theta(u)-l_{\theta^*}(u)\right]^2=\EE_P\left[(-\gamma)^{2\alpha}(x^\top V^*V^{*\top} x -x^\top VV^\top x)^2\right]  \\
  & = \beta \EE_{x\sim P_F}\left[x^\top \Delta xx^\top \Delta x\right]+(1-\beta)\gamma^2\EE_{x\sim P_B} \left[x^\top \Delta xx^\top \Delta x\right]\\
  & = \beta \left[\tr(\Delta C_F)^2+2\tr(\Delta C_F\Delta C_F)\right]+(1-\beta)\gamma^2\left[\tr(\Delta C_B)^2+2\tr(\Delta C_B\Delta C_B)\right].
\end{align*}
Rewrite $m$ in a similar form:
\begin{align*}
    m(\theta,\theta^*)& = R(\theta)-R(\theta^*)= \EE_P(l_\theta-l_{\theta^*})\\
    & = \EE_P\left[(-\gamma)^\alpha x^\top \Delta x\right]\\
    & = \beta\tr(\Delta C_F) +(1-\beta)\gamma\tr(\Delta C_B). 
\end{align*}
Now we can calculate $v(\theta,\theta^*)$:
\begin{align*}
    &v(\theta,\theta^*) = \EE_p\left[l_\theta(u)-l_{\theta^*}(u)\right]^2-m(\theta,\theta^*)^2\\
    & = \beta \left[\tr(\Delta  C_F)^2+2\tr(\Delta C_F\Delta C_F)\right]+(1-\beta)\gamma^2\left[\tr(\Delta  C_B)^2+2\tr(\Delta  C_B\Delta C_B)\right]\\
    &~~~ -\left[\beta\tr(\Delta C_F) +(1-\beta)\gamma\tr(\Delta C_B)\right]^2\\
    & = (\beta-\beta^2)\tr(\Delta C_F)^2+\gamma^2(\beta-\beta^2)\tr(\Delta C_B)^2+2\beta\tr(\Delta C_F\Delta C_F)\\
    &~~~~+2\gamma^2(1-\beta)\tr(\Delta C_B\Delta C_B)-2\gamma\beta(1-\beta)\tr(\Delta C_F)\tr(\Delta C_B).
\end{align*}

\begin{proof}
Observe that $\|\Delta \|\lambda^F_D\leq|\tr(\Delta C_F)|\leq \|\Delta \|\lambda^F_1$, where $\lambda^F_1$ and $\lambda^F_D$ are the largest and smallest eigenvalue of $C_F$. Then by the above calculation, we know that
$$m(\theta,\theta^*)\sim\|VV^\top-V^*V^{*\top}\|,~~v(\theta,\theta^*)\sim\|VV^\top-V^*V^{*\top}\|^2.$$
When $n$ is sufficiently large, there exists constant $c$ such that
$$\left\{\theta:\|VV^\top-V^*V^{*\top}\|\leq cn^{-1/2}\right\}\subset \left\{\theta:m(\theta,\theta^*)\vee v(\theta,\theta^*)\leq n^{-1/2}\right\}.$$
So it suffices to check the prior $\Pi$ assigns enough mass around $V^*$ w.r.t. the operator norm. 
Recall that the Riemannian volume measure on $\Gr(D,d)$, denoted by $\mathcal{H}$, is  $O(n)$ invariant, also known as the Haar measure, while the distance on $\Gr(D,d)$ is given by $$d(V_1,V_2) = \|V_1V_1^\top-V_2V_2^\top\|.$$ Denote the ball centered at $V$ with radius $r$ w.r.t. this distance by $B(V,r)$, then there exists constant $c$ such that
$$\frac{1}{c}r^{d(D-d)}\leq\mathcal{H}(B(V,r))\leq cr^{d(D-d)},~~\forall V\in \Gr(D,d),~r>0.$$

When the prior $\Pi$ is uniform w.r.t. the Haar measure, \begin{align*}
\Pi\left\{\theta:m(\theta,\theta^*)\vee v(\theta,\theta^*)\leq n^{-1/2}\right\}\gtrsim\Pi\left\{\theta:\|\Delta \|_o\leq cn^{-1/2}\right\}\gtrsim  \left(n^{-1/2}\right)^{d(D-d)},
\end{align*}
where $d(D-d)=\dim(\Gr(D,d))$, so by \cite[Theorem 3.3]{syring2020gibbs}, the posterior contraction rate of the Gibbs posterior is $n^{-1/2}$, which is optimal.
\end{proof}

\subsection{Proof of Theorem \ref{thm:PCPCArate}}
\label{proof:PCPCArate}
As in the proof of Theorem \ref{thm:CPCArate}, the loss is of sub-exponential type.
Then it suffices to check that the prior $\Pi$ satisfies Lemma \autoref{lem:syring2020gibbs}. To do so, we first calculate $m$ and $v$. Let $\delta = \log|A|-\log|A^*|$ and $\Delta = A^{-1}-{A^*}^{-1}$, then 
\begin{align*}
    m(\theta,\theta^*) &=R(\theta)-R(\theta^*)=\frac{\beta-(1-\beta)\gamma}{2}(\log|A|-\log|A^*|)+\frac{\tr((A^{-1}-A^{*-1})C)}{2}\\
    & = \frac{\beta-(1-\beta)\gamma}{2}\delta+\frac{1}{2}\tr(\Delta C)=O(\sqrt{\delta^2+\|\Delta\|^2}).
\end{align*}
Now, to calculate $v$, we have
\begin{align*}
  &\EE_p\left[l_\theta(u)-l_{\theta^*}(u)\right]^2=\EE_P\left[(-\gamma)^{2\alpha}\left(\frac{1}{2}(\log|A|-\log|A^*|)+\frac{1}{2}x^\top(A^{-1}-{A^{*}}^{-1})x\right)^2\right]  \\
  & = \frac{\beta}{4} \EE_{x\sim P_F}\left[\delta^2+2\delta x^\top\Delta x+x^\top \Delta x x^\top \Delta x\right]+\frac{(1-\beta)\gamma^2}{4}\EE_{x\sim P_B} \left[\delta^2+2\delta x^\top\Delta x+x^\top \Delta x x^\top \Delta x\right]\\
  & = \frac{\beta}{4} \left[\delta^2+2\delta \tr(\Delta C_F)+\tr(\Delta C_F)^2+2\tr(\Delta C_F\Delta C_F)\right]\\
  & ~~~+\frac{(1-\beta)\gamma^2}{4} \left[\delta^2+2\delta \tr(\Delta C_B)+\tr(\Delta C_B)^2+2\tr(\Delta C_B\Delta C_B)\right]\\
  &=O(\delta^2+\|\Delta\|^2). 
\end{align*}
As a result,
\begin{align*}
    &v(\theta,\theta^*) = \EE_p\left[l_\theta(u)-l_{\theta^*}(u)\right]^2-m(\theta,\theta^*)^2=O( \delta^2+\|\Delta\|^2). 
\end{align*}

\begin{proof}
By the assumption that $\sigma^2\geq \sigma_0^2$, all eigenvalues of $A=WW^\top+\sigma^2\Id_D$ and $A^*=W^*W^{*\top}+\sigma^{*2}\Id_D$ are lower-bounded by $\sigma_0^2$. As a result, both $\log|\cdot|$ and $\tr(\cdot^{-1}C)$ are Lipschitz, and both $m$ and $v$ can be bounded by the distance between parameters. Then by the above calculation, there exists a $c$ such that
$$m(\theta,\theta^*)\leq c\sqrt{\|WW^\top-W^*W^{*\top}\|^2+(\sigma^2-\sigma^{*2})^2},$$
$$v(\theta,\theta^*)\leq c\left(\|WW^\top-W^*W^{*\top}\|^2+(\sigma^2-\sigma^{*2})^2\right).$$
When $n$ is sufficiently large, there exists a constant $c$ such that
$$\left\{\theta:\sqrt{\|WW^\top-W^*W^{*\top}\|^2+(\sigma^2-\sigma^{*2})^2}\leq cn^{-1/2}\right\}\subset \left\{\theta:m(\theta,\theta^*)\vee v(\theta,\theta^*)\leq n^{-1/2}\right\}.$$
So it suffices to check that the prior $\Pi$ assigns enough mass around $(V^*,\sigma^{*2})$ w.r.t. the product measure. Denote the ball centered at $(V,\sigma^2)$ with radius $r$ w.r.t. this distance by $B(V,\sigma^2,r)$, then there exists a constant $c$ such that
$$\frac{1}{c}r^{Dd+1}\leq\mathrm{Vol}(B(V,\sigma^2,r))\leq cr^{Dd+1},~~\forall (W,\sigma^2)\in\RR^{D\times d}\times [\sigma_0^2,\infty),~r>0.$$

When the prior $\Pi$ is uniform w.r.t. the Haar measure, \begin{align*}
&~~~~\Pi\left\{\theta:m(\theta,\theta^*)\vee v(\theta,\theta^*)\leq n^{-1/2}\right\}\\
&\gtrsim\Pi\left\{\theta:\sqrt{\|WW^\top-W^*W^{*\top}\|^2+(\sigma^2-\sigma^{*2})^2}\leq cn^{-1/2}\right\}\gtrsim  \left(n^{-1/2}\right)^{dD+1}.
\end{align*}
We conclude that, by \cite[Theorem 3.3]{syring2020gibbs}, the posterior contraction rate of the Gibbs posterior is $n^{-1/2}$, which is optimal.
\end{proof}

\end{document}